\documentclass[a4paper,USenglish,cleveref,autoref,nolineno,thm-restate]{socg-lipics-v2021}

\newcommand{\old}[1]{{}}
\newcommand{\eps}{\varepsilon}
\newcommand{\diam}{\mathrm{diam}}
\newcommand{\dir}{\mathrm{dir}}
\newcommand{\per}{\mathrm{per}}
\newcommand{\hper}{\mathrm{hper}}
\newcommand{\vper}{\mathrm{vper}}
\newcommand{\slope}{\mathrm{slope}}

\newcommand{\wdth}{\text{width}}
\newcommand{\hght}{\text{height}}

\newcommand{\MST}{\text{MST}}

\newtheorem{question}[theorem]{Question}

\hypersetup{
hidelinks,
colorlinks=true,
citecolor=[rgb]{0 0 1},
linkcolor=[rgb]{0 0 1},
urlcolor=[rgb]{0 0 1}
}
\hideLIPIcs
\nolinenumbers

\title{Light Euclidean Steiner Spanners in the Plane}

\titlerunning{Light Euclidean Steiner Spanners in the Plane}

\author{Sujoy Bhore}{Universit\'e Libre de Bruxelles, Brussels, Belgium.}{sujoy.bhore@gmail.com}{0000-0003-0104-1659} {Research on this paper was supported by the Fonds de la Recherche Scientifique-FNRS under Grant no MISU F 6001.}
\author{Csaba D. T\'oth}{California State University Northridge, Los Angeles, CA; and Tufts University, Medford, MA, USA.}{csaba.toth@csun.edu}{0000-0002-8769-3190}
{Research on this paper was partially supported by the NSF award DMS-1800734.}

\authorrunning{S.~Bhore, Cs.~D.~T\'oth}

\Copyright{Sujoy Bhore and Csaba D. T\'oth}



\EventEditors{Kevin Buchin and \'{E}ric Colin de Verdi\`{e}re}
\EventNoEds{2}
\EventLongTitle{37th International Symposium on Computational Geometry (SoCG 2021)}
\EventShortTitle{SoCG 2021}
\EventAcronym{SoCG}
\EventYear{2021}
\EventDate{June 7--11, 2021}
\EventLocation{Buffalo, NY, USA}
\EventLogo{socg-logo.pdf}
\SeriesVolume{189}
\ArticleNo{31}

\begin{document}

\maketitle

\begin{abstract}
Lightness is a fundamental parameter for Euclidean spanners; it is the ratio of the spanner weight to the weight of the minimum spanning tree of a finite set of points in $\mathbb{R}^d$. In a recent breakthrough, Le and Solomon (2019) established the precise dependencies on $\eps>0$ and $d\in \mathbb{N}$ of the minimum lightness of $(1+\eps)$-spanners, and observed that additional Steiner points can substantially improve the lightness. Le and Solomon (2020) constructed Steiner $(1+\eps)$-spanners of lightness $O(\eps^{-1}\log\Delta)$ in the plane, where  $\Delta\geq \Omega(\sqrt{n})$ is the \emph{spread} of the point set, defined as the ratio between the maximum and minimum distance between a pair of points.
They also constructed spanners of lightness $\tilde{O}(\eps^{-(d+1)/2})$ in dimensions $d\geq 3$. Recently, Bhore and T\'{o}th (2020) established a lower bound of $\Omega(\eps^{-d/2})$ for the lightness of Steiner $(1+\eps)$-spanners in $\mathbb{R}^d$, for $d\ge 2$.
The central open problem in this area is
to close the gap between the lower and upper bounds in all dimensions $d\geq 2$.

In this work, we show that for every finite set of points in the plane and every $\eps>0$, there exists a Euclidean Steiner $(1+\eps)$-spanner of lightness $O(\eps^{-1})$; this matches the lower bound for $d=2$. We generalize the notion of shallow light trees, which may be of independent interest, and use directional spanners and a modified window partitioning scheme to achieve a tight weight analysis.
\end{abstract}

\section{Introduction}
\label{sec:intro}

Given an edge-weighted graph $G$, a spanner is a subgraph $H$ of $G$ that preserves the length of the shortest paths in $G$ up to some amount of multiplicative or additive distortion. Formally, a subgraph $H$ of a given edge-weighted graph $G$ is a $t$-\emph{spanner}, for some $t\ge 1$, if for every $pq\in \binom{V(G)}{2}$ we have $d_H(p,q)\leq t\cdot d_G(p,q)$, where $d_G(p,q)$ denotes the length of
the shortest path in $G$.
The parameter $t$ is called the \emph{stretch factor} of the spanner. Graph spanners were introduced by Peleg and Sch\"{a}ffer~\cite{peleg1989graph}, and since then it has turned out to be a  fundamental graph structure with numerous applications in the field of distributed systems and communication, distributed queuing protocol, compact routing schemes, etc.; see~\cite{demmer1998arrow,
herlihy2001competitive, PelegU89a, peleg1989trade}.
For edge-weighted graphs, a natural parameter is the \emph{lightness} of a spanner, that is associated with the total weight of the spanner. The \emph{lightness} of a spanner $H$ of an input graph $G$, is the  ratio $w(H)/w(\MST)$ between the total weight of $H$ and the weight of a minimum spanning tree $(\MST)$ of $G$. Note that, since a spanner $H$ is a connected graph, the trivial lower bound for lightness is $1$.

In geometric settings, a $t$-spanner for a finite set $S$ of points in $\mathbb{R}^d$ is a subgraph of the underlying complete graph $G=(S,\binom{S}{2})$, that preserves the pairwise Euclidean distances between points in $S$ to within a factor of $t$, which is the \emph{stretch factor}. The edge weights of $G$ are the Euclidean distances between the vertices. Chew~\cite{Chew86, Chew89} initiated the study of Euclidean spanners in 1986, and showed that for a set of $n$ points in $\mathbb{R}^2$, there exists a spanner with $O(n)$ edges and constant stretch factor.
Since then a large body of research has been devoted to Euclidean spanners due to its many applications across domains, such as, topology control in wireless networks~\cite{schindelhauer2007geometric}, efficient regression in metric spaces~\cite{gottlieb2017efficient},
approximate distance oracles~\cite{gudmundsson2008approximate}, and others. Moreover, Rao and Smith~\cite{rao1998approximating} showed the relevance of Euclidean spanners in the context of other fundamental geometric \textsf{NP}-hard problems, e.g., Euclidean traveling salesman problem and Euclidean minimum Steiner tree problem. Many different spanner construction approaches have been developed for Euclidean spanners over the years, that each found further applications in geometric optimization, such as spanners based on
well-separated pair decomposition (WSPD)~\cite{callahan1993optimal, GudmundssonLN02},
skip-lists~\cite{arya1994randomized},
path-greedy and gap-greedy approaches~\cite{althofer1993sparse, arya1997efficient},
locality-sensitive orderings~\cite{ChanHJ20}, and more. We refer to the book by Narasimhan and Smid~\cite{narasimhan2007geometric} and the survey by Bose and Smid~\cite{BoseS13} for a summary of results and techniques on Euclidean spanners up to 2013.

Lightness and sparsity are two natural parameters for Euclidean spanners. For a set $S$ of points in $\mathbb{R}^d$,  the lightness is the ratio of the spanner weight (i.e., the sum of all edge weights) to the weight of the Euclidean minimum spanning tree
$\MST(S)$. Then, the \emph{sparsity} of a spanner on $S$ is the ratio of its size to the size of a spanning tree. Classical results show that when the dimension $d\in \mathbb{N}$ and $\varepsilon > 0$ are constant, every set $S$ of $n$ points in $d$-space admits an
$(1 +\varepsilon)$-spanners with $O(n)$ edges and weight proportional to that of the Euclidean \textsc{MST} of $S$. We refer to a series of spanners constructions for a comprehensible understanding of sparse spanners~\cite{Chew89, Clarkson87,
keil1988approximating, keil1992classes,
ruppert1991approximating,
yao1982constructing}.

Of particular interest, we elaborate on the \emph{lightness} aspect of Euclidean spanners. Das et al.~\cite{das1993optimally} showed that the \emph{greedy-spanner} (cf.~\cite{althofer1993sparse}) has constant lightness in $\mathbb{R}^3$. This was generalized later to $\mathbb{R}^d$, for all $d\in \mathbb{N}$, by Das et al.~\cite{narasimhan1995new}. However the dependencies on $\eps$ and $d$ have not been addressed. Rao and Smith~\cite{rao1998approximating} showed
that the greedy spanner has lightness $\eps^{-O(d)}$ in $\mathbb{R}^d$ for every constant $d$, and asked what is the best possible constant in the exponent. A complete proof for the existence of a $(1+\eps)$-spanner with lightness $O(\eps^{-2d})$ is in the book on geometric spanners~\cite{narasimhan2007geometric}. Gao et al.~\cite{gao2006deformable} considered $(1+\eps)$-spanners in doubling metrics, and showed that every finite set of $n$ points in doubling dimension $d$ has a spanner of sparsity $\eps^{-O(d)}$.
In 2015, Gottlieb~\cite{gottlieb2015light} showed that a metric of doubling dimension $d$ has a spanner of lightness $(d/\eps)^{O(d)}$, which improved the $O(\log n)$ lightness bound of
Smid~\cite{smid2009weak}. Recently, Borradaile et al.~\cite{borradaile2019greedy} showed that the greedy $(1+\eps)$-spanner of a finite metric space of doubling dimension $d$ has lightness $\eps^{-O(d)}$. Le and Solomon~\cite{le2019truly} established the dependencies of $\varepsilon$ in the \emph{lightness} and \emph{sparsity} bounds of Euclidean $(1+\eps)$-spanners up to polylogerithmic factors.
They constructed, for every $\eps>0$ and constant $d\in \mathbb{N}$, a set $S$ of $n$ points in $\mathbb{R}^d$ for which any $(1+\eps)$-spanner must have lightness $\Omega(\eps^{-d})$ and sparsity $\Omega(\eps^{-d+1})$,
whenever $\eps = \Omega(n^{-1/(d-1)})$. Moreover, they showed that the greedy $(1+\eps)$-spanner in $\mathbb{R}^d$ has lightness $O(\eps^{-d}\log \eps^{-1})$. Several classical constructions are known to achieve sparsity $O(\eps^{-d+1})$ in $\mathbb{R}^d$~\cite{althofer1993sparse,le2019truly,narasimhan2007geometric}.

\subparagraph{Steiner Spanners.} \emph{Steiner points} are additional vertices in a network that are not part of the input, and a $t$-spanner must achieve stretch factor $t$ only between pairs of the input points in $S$.
Le and Solomon~\cite{le2019truly} observed that it is possible to use Steiner points to bypass the lower bounds and substantially improve the bounds for \emph{lightness} and \emph{sparsity}
of Euclidean $(1+\eps)$-spanners. For minimum sparsity, they gave an upper bound of $O(\eps^{(1-d)/2})$ for $d$-space and a lower bound of $\Omega(\eps^{-1/2}/\log\eps^{-1})$. For minimum lightness, they gave a lower bound of $\Omega(\eps^{-1}/\log\eps^{-1})$, for points in the plane ($d=2$)~\cite{le2019truly}. In a subsequent work~\cite{le2020light}, they have constructed Steiner $(1+\eps)$-spanners of lightness $O(\eps^{-1}\log\Delta)$ in the plane, where $\Delta$ is the \emph{spread} of the point set, defined as the ratio between the maximum and minimum distance between a pair of points. In particular, $\log \Delta\in \Omega(\log n)$ in doubling metrics.

Recently, Bhore and T\'{o}th~\cite{BT-less-21} established a lower bound of $\Omega(\eps^{-d/2})$ for the lightness of Steiner $(1+\eps)$-spanners in Euclidean $d$-space for all $d\ge 2$. Moreover, for points in the plane, they established an upper bound of $O(\eps^{-1}\log n)$. In~\cite{le2020unified},
Le and Solomon constructed spanners of lightness $\tilde{O}(\eps^{-(d+1)/2})$ in dimensions $d\geq 3$,
nearly matching the lower bound $\Omega(\eps^{-d/2})$, for $d\ge 3$. The central open problem in this area is to close the gap between the lower and upper bounds of lightness, in all dimensions $d\geq 2$.

\begin{question}\label{ques-sp-plane}
Do there exist Euclidean Steiner $(1+\eps)$-spanners for a finite set of points in $\mathbb{R}^d$, of lightness $O(\eps^{-d/2})$, for any $d\ge 2$?
\end{question}

Bounding the \emph{lightness} of Euclidean spanners is often harder than bounding the \emph{sparsity}, as also noted by Le and Solomon~\cite{le2020light}. Several works portrayed the difficulties of constructing light spanners in Euclidean spaces, doubling metrics, as well as on other weighted graphs;
see~\cite{althofer1993sparse, borradaile2019greedy,
chechik2018near, elkin2015light,
gottlieb2015light,
le2019truly, smid2009weak,
narasimhan1995new, rao1998approximating}.
A delicate aspect of the problem is to find suitable locations for \emph{Steiner points}.
Recent results on Steiner spanners~\cite{BT-less-21, le2019truly, le2020light, le2020unified} suggest that highly nontrivial insights are required to argue the upper bounds for Steiner spanners, and they tend to be even more intricate than their non-Steiner counterpart.

\subparagraph{Related Previous Work.}
Steiner points were used in several occasions to improve the overall weight of a network.
Previously, Elkin and Solomon~\cite{elkin2015steiner} and Solomon~\cite{Solomon15} showed that Steiner points can improve the weight of the network in the single-source setting. In particular, they introduced the so-called \emph{shallow-light trees} (\textsf{SLT}), that is a single-source spanning tree that concurrently approximates a shortest-path tree (between the source and all other points) and a minimum spanning tree (for the total weight). They proved that Steiner points help to obtain exponential improvement on the lightness of \textsf{SLT}s in a general metric space~\cite{elkin2015steiner}, and quadratic improvement on the lightness in Euclidean spaces~\cite{Solomon15}.

\subparagraph{Our Contribution.}
In this work, we show that for every finite set of points in the plane and every $\eps>0$, there exists a Euclidean Steiner $(1+\eps)$-spanner of lightness $O(\eps^{-1})$ (Theorem~\ref{thm:UB}). This matches the lower bound for $d=2$, and thereby closes the gap between lower and upper bounds of lightness for Euclidean $(1+\eps)$-spanners in $\mathbb{R}^2$.

On the one hand, without Steiner points, the greedy spanner in Euclidean plane has lightness $\tilde{O}(\eps^{-2})$, which is the best possible up to lower-order terms~\cite{le2019truly}.  On the other hand, with Steiner points, recent constructions achieved linear dependence on $\eps^{-1}$, while loosing the independence from $n$; see~\cite{BT-less-21, le2020light}.
Our result is the first that constructs Steiner spanners with sub-quadratic dependence on $\eps^{-1}$ without any dependence on $n$ or any assumption on the point set; in fact our result achieves the optimal dependence on $\eps$.

\begin{theorem}\label{thm:UB}
For every finite point sets $S\subset \mathbb{R}^2$ and $\eps>0$, there exists a Euclidean Steiner $(1+\eps)$-spanner of weight $O(\frac{1}{\eps}\,\|\MST(S)\|)$.
\end{theorem}

\subparagraph{Outline.}
We review previous results on angle-bounded paths, \textsf{SLT}s, and window partitions that we use in our construction (Section~\ref{sec:pre}). The tight bound in Theorem~\ref{thm:UB} relies on three new ideas, which may be of independent interest: First, we generalize Solomon's \textsf{SLT}s to points on a staircase path (Section~\ref{sec:SLT}). Second, we reduce the proof of Theorem~\ref{thm:UB} to the construction of ``directional'' spanners, in each of $\Theta(\eps^{-1/2})$ directions, where it is enough to establish the stretch factor $1+\eps$ for point pairs $s,t\in S$ where $\mathrm{dir}(st)$ is in an interval of size $\sqrt{\eps}$ (Section~\ref{sec:redux}). Combining the first two ideas, we show how to construct light directional spanners for points on a staircase path (Section~\ref{sec:staircases}).
In each direction, we start with 
a rectilinear MST of $S$, and augment it into a directional spanner.
We modify the classical window partition of a rectilinear polygon into histograms by replacing vertical edges with angle-bounded paths; this is the final piece of the puzzle. Near-vertical paths (with slopes
$\pm \, \eps^{-1/2}$) allow sufficient flexibility to reduce the weight of a histogram subdivision, and we can construct directional $(1+\eps)$-spanners for each face of such a subdivision (Section~\ref{sec:hist}).

\section{Preliminaries}
\label{sec:pre}

The \emph{direction} of a line segment $ab$ in the plane, denoted $\mathrm{dir}(ab)$, is the minimum counterclockwise angle $\alpha\in [0,\pi)$ that rotates the $x$-axis to be parallel to $ab$. The set of possible directions $[0,\pi)$ is homeomorphic to the unit circle $\mathbb{S}^1$, and an interval $(\alpha,\beta)$ of directions corresponds to the counterclockwise arc of $\mathbb{S}^1$ from $\alpha\pmod{\pi}$ to $\beta\pmod{\pi}$.

\subparagraph{Angle-Bounded Paths.}
For $\delta\in (0,\pi/2]$, a polygonal path $(v_0,\ldots, v_m)$ is \emph{$(\theta\pm \delta)$-angle-bounded} if the direction of every segment $v_{i-1}v_i$ is in the interval $[\theta-\delta,\theta+\delta]$; see Fig.~\ref{fig:1}(a).
Borradaile and Eppstein~\cite[Lemma~5]{BorradaileE15} observed that the
weight of a $(\theta\pm \delta)$-angle-bounded $st$-path is at most $(1+O(\delta^2))\|st\|$.
We prove this observation in a more precise form. The quadratic growth rate in $\delta$ is due to the Taylor estimate $\sec(x)=\frac{1}{\cos(x)}\leq 1+x^2$ for $x\leq \frac{\pi}{4}$.

\begin{lemma}\label{lem:angle2}
Let $a,b\in \mathbb{R}^2$ and let $P=v_0 v_1 \ldots v_m$ be an $ab$-path such that $P$ is monotonic in direction $\overrightarrow{ab}$ and $|\mathrm{dir}(v_{i-1}v_i)-\mathrm{dir}(ab)|\leq \delta\leq \frac{\pi}{4}$, for $i=1,\ldots ,m$. Then $\|P\|\leq (1+\delta^2)\|ab\|$.
\end{lemma}
\begin{proof}
For $i=0,\ldots , m$, let $u_i$ be the orthogonal projection of $v_i$ to the line $ab$, and let $\alpha_i=\mathrm{dir}(v_{i-1}v_i)-\mathrm{dir}(ab)$. Then
$\|ab\|=\sum_{i=1}^m \|u_{i-1}u_i\| =\sum_{i=1}^m \|v_{i-1}v_i\| \sec \angle_i \leq \|P\|\sec \delta\leq (1+\delta^2)\|P\|$, as claimed.
\end{proof}

\begin{figure}[htbp]
 \centering
 \includegraphics[width=0.98\textwidth]{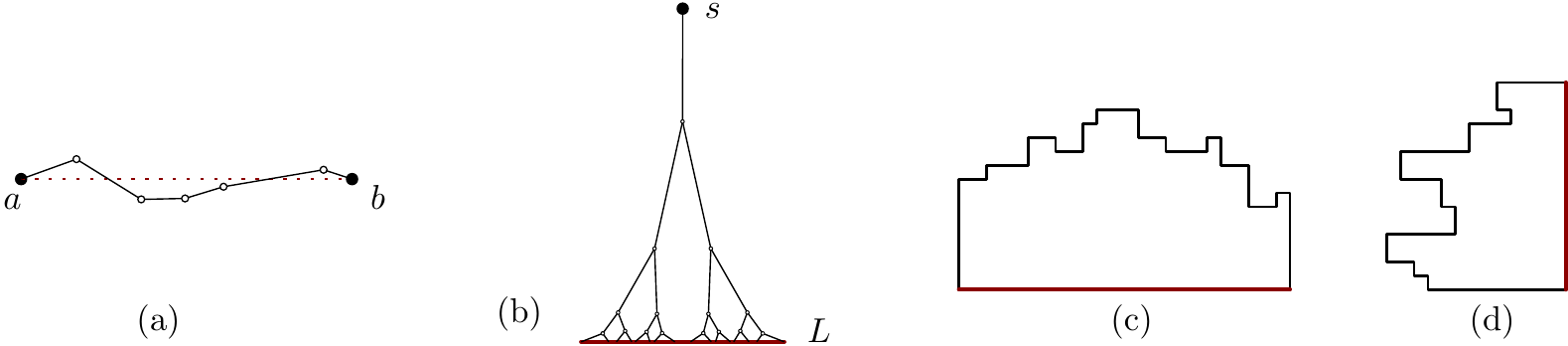}
 \caption{(a) A $(0\pm\delta)$-angle-bounded path.
(b) A shallow-light tree between a source $s$ and a horizontal line segment $L$.
(c)--(d) An $x$- and a $y$-monotone histogram.}
    \label{fig:1}
\end{figure}

\subparagraph{Shallow-Light Trees.}
Shallow-light trees (\textsf{SLT}) were introduced by Awerbuch et al.~\cite{AwerbuchBP90}
and Khuller et al.~\cite{KhullerRY93}. Given a source $s$ and a point set $S$ in a metric space, an \emph{$(\alpha,\beta)$}-\textsf{SLT} is a Steiner tree rooted at $s$ that contains a path of weight at most $\alpha \,\|ab\|$ between the \emph{source} $s$ and any point $t\in S$, and has weight at most $\beta\,\|\MST(S)\|$. We build on the following basic variant of \textsf{SLT} between a source $s$ and a set $S$ of collinear points in the plane; see Fig.~\ref{fig:1}(b).

\begin{lemma}[Solomon~{\cite[Section~2.1]{Solomon15}}]\label{lem:shallow}
Let $0<\eps<1$, let $s=(0,\eps^{-1/2})$ be a point on the $y$-axis, and let $S$ be a set of points in the line segment $L=[-\frac12,\frac12]\times \{0\}$ in the $x$-axis. Then there exists a geometric graph of weight $O(\eps^{-1/2})$ that contains, for every point $t\in L$, an $st$-path $P_{st}$ with $\|P_{st}\|\leq (1+\eps)\,\|st\|$.
\end{lemma}

We note that the weight analysis of the $st$-path $P_{st}$ in an \textsf{SLT} does not use  angle-boundedness. In particular, an \textsf{SLT} may contain short edges of arbitrary directions, but the directions of long edges are nearly to vertical.
In Section~\ref{sec:SLT} below, we generalize the shallow-light trees to obtain
$(1+\eps)$-spanners between points on two staircase paths.

\subparagraph{Stretch Factor of $1+\eps$ Versus $1+O(\eps)$.}
In the geometric spanners we construct, an $st$-path may comprise $O(1)$ subpaths, each of which is angle-bounded or contained in an \textsf{SLT}. For the ease of presentation, we typically establish a stretch factor of $1+O(\eps)$ in our proofs. It is understood that $1+\eps$ can be achieved by a suitable scaling of the constant coefficients.

\subparagraph{Histogram Decomposition.}
A path in the plane is \emph{$x$-monotone} (resp., \emph{$y$-monotone}) if its intersection with
every vertical (resp., horizontal) line is connected.
A \emph{histogram} is a rectilinear simple polygon bounded by an axis-parallel line segment
and an $x$- or $y$-monotone path; see Fig.~\ref{fig:1}(c--d).
It is well known that every rectilinear simple polygon $P$ can be subdivided into histograms (faces) such that every axis-parallel line segment in $P$ intersects (\emph{stabs}) at most three histograms~\cite{Edelsbrunner1984167,Levcopoulos}; such a subdivision is also called a
\emph{window partition}~\cite{Link00,Suri90} of $P$, and can be computed in $O(n\log n)$ time if $P$ has $n$ vertices. The stabbing property implies that the total perimeter of the histograms in such a subdivision is $O(\per(P))$, where $\per(P)$ denotes the perimeter of $P$.

Dumitrescu and T\'oth~\cite{DumitrescuT09} showed that for a finite point set $S\subset \mathbb{R}^2$, one can refine the window partition, while increasing the weight by a constant factor, to construct a graph with constant geometric dilation. The \emph{geometric dilation} of a geometric graph $G$ is $\sup_{a,b\in G}d_G(a,b)/\|ab\|$, where $d_G(a,b)$ denotes the Euclidean length of a shortest path in $G$, and the supremum is taken over all point pairs $\{a,b\}$ at vertices and along edges of $G$. We follow a similar approach here, but we construct a subdivision of ``modified'' histograms (defined in Section~\ref{sec:hist}), where the vertical edges are replaced by angle-bounded paths.

\section{Generalized Shallow Light Trees}
\label{sec:SLT}

In Section~\ref{ssec:stairs}, we generalize Lemma~\ref{lem:shallow}, and construct shallow-light trees between a source $s$ and points on an $x$- and $y$-monotone rectilinear path $L$, which is called a \emph{staircase path}. In Section~\ref{ssec:combination}, we show how to combine two shallow-light trees to obtain a spanner between point pairs on two staircase paths.

\subsection{Single Source and Staircase Chain}
\label{ssec:stairs}

We present a new, slightly modified proof for Solomon's result on \textsf{SLT}s between a single source $s$ and a horizontal line segment, and then adapt the modified proof to obtain a \textsf{SLT} between $s$ and an $x$- and $y$-monotone polygonal chain. In the proof below, we use the Taylor estimates $\cos x\geq 1-x^2/2$ and $\sin x\geq x/2$ for $x\leq \pi/3$.

\begin{figure}[htbp]
 \centering
 \includegraphics[width=0.75\textwidth]{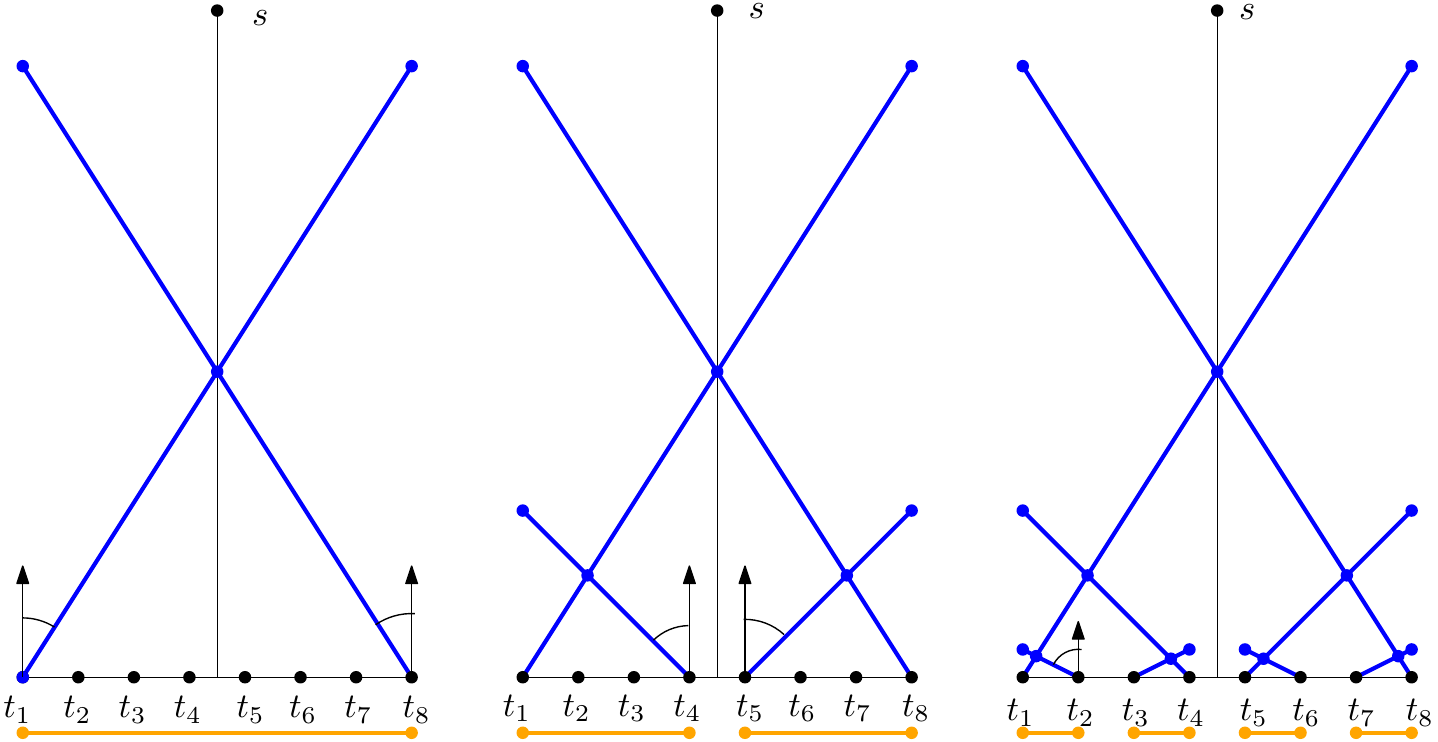}
 \caption{The segments added to graph $G$ at level $j=0,1,2$ for $m=2^3=8$ points.
 The intervals $[t_a,t_b]$ at level $j$ are indicated below the line $L$.}
\label{fig:SLT}
\end{figure}

\begin{proof}[Alternative proof for Lemma~\ref{lem:shallow}]
Assume w.l.o.g.\ that $\eps=2^{-k}$ for $k\in \mathbb{N}$. Let $T=\{t_i: i=1,\ldots , 2^{k+1}\}$ be $2^{k+1}$ points on the line segment $L=[-\frac12,\frac12]\times \{0\}$ with uniform $1/(2^{k+1}-1)< \eps$ spacing between consecutive points.
Consider the standard binary partition of $\{1,\ldots , 2^{k+1}\}$ into intervals, associated with a binary tree: At level 0, the root corresponds to the interval $[1,2^{k+1}]$ of all $2^{k+1}$ integers. At level $j$, we have intervals $[i\cdot 2^{k-j}+1,(i+1)\cdot 2^{k-j}]$ for $i=0,\ldots  ,2^j-1$. Note that if a point $q$ is the left (resp., right) endpoint of an interval at a level $j$, then $q$ is the left (resp., right) endpoint of all descendant intervals that contains $q$.

For every $q\in \{1,\ldots , 2^{k+1}\}$, we define a line segment $\ell_q$ with one endpoint at $t_q$: Let $j\geq 0$ be the smallest level such that $q$ is an endpoint of some interval $I_q$ at level $j$. Let $T_q$ be the line segment along the $x$-axis spanned by the points corresponding to $I_q$. If $q$ is the left (resp., right) endpoint of $I_q$, then let $\ell_q$ be the line segment of direction $\frac{\pi}{2}-2^{(j-k)/2}$ (resp., $\frac{\pi}{2}+2^{(j-k)/2})$
such that its orthogonal projection to the $x$-axis is $T_q$; see Fig.~\ref{fig:SLT}.
Note that for $j=0$, we use directions $\frac{\pi}{2}\pm 2^{-k/2}=\frac{\pi}{2}\pm \sqrt{\eps}$.
Let $G$ be the union of segments $\ell_q$ for $q=1,\ldots , 2^{k+1}$,
the horizontal segment $L$, and the vertical segment from $s$ to the origin.

\smallskip\noindent\emph{Lightness analysis.}
We show that $\|G\|=O(\eps^{-1/2})$. We have $\|L\|=1$, and the length of the vertical segment between $s$ and the origin is $\eps^{-1/2}$. At level $j$ of the binary tree, we construct $2^j$ segments $\ell$, each of length $\|\ell\|\leq 2^{-j}/\sin(2^{(j-k)/2})\leq 2\cdot 2^{(k-3j)/2}$. Summation over all levels yields
$\sum_{j=0}^{k} 2^{j}\cdot 2\cdot 2^{(k-3j)/2}= 2^{k/2}\cdot 2\cdot \sum_{j=0}^k 2^{-j/2}= O(2^{k/2}) = O(\eps^{-1/2})$.

\smallskip\noindent\emph{Source-stretch analysis.}
We show that $G$ contains an $st_q$-path of length $(1+O(\eps))\|st_q\|$ for all $q=1,\ldots, 2^{k+1}$. First note that $\|st_q\|\geq \eps^{-1/2}$, as the distance between $s$ and $L$ is $\eps^{-1/2}$. For each interval $[t_a,t_b]$ in the binary tree, $\ell_a$ and $\ell_b$ have positive and negative slopes, respectively, and so they cross above the interval $[t_a,t_b]$. Consequently, for every point $t_q$, the union of the $k+1$ segments corresponding to the intervals that contain $t_q$ must contain a $y$-monotonically increasing path $P_q$ from $t_q$ to $s$. The $y$-projection of this path has length $\eps^{-1/2}$. Consider one edge $e$ of $P_q$ along a segment $\ell$ at level $j$, which has direction $\frac{\pi}{2}\pm \alpha=\frac{\pi}{2}\pm 2^{(j-k)/2}$. Then the difference between the length of $e$ and the $y$-projection of $e$ is
$\|e\|(1-\cos \alpha)
\leq \|\ell\|(1-\cos\alpha)
\leq 2^{-j} \frac{1-\cos\alpha}{\sin \alpha}
\leq 2^{-j} \frac{\alpha^2/2}{\alpha/2}
= 2^{-j} \alpha = 2^{-j}\cdot 2^{(j-k)/2} =2^{-(j+k)/2}$.
Since $P_q$ contains at most one edge in each level, summation over all edges of $P_q$ yields
$\sum_{j=0}^{k} 2^{-(j+k)/2} = 2^{-k/2} \sum_{j=0}^{k} 2^{-j/2} = O(\eps^{1/2})\leq \|st_q\|\cdot O(\eps)$.

Finally, for an arbitrary point $t\in L$, we have $\|st\|\geq \eps^{-1/2}$, and $G$ contains an $st$-path that consists of an $st_q$-path from $s$ to the point $t_q$ closest to $t$, followed by the horizontal segment $t_qt$ of weight $\|t_qt\|< 1/2^k\leq \eps$. The total weight of this path is $(1+O(\eps))\|st\|$.
After suitable scaling of the constant coefficients, $G$ contains a path of weight at most $(1+\eps)\|st\|$ for any $t\in L$, as required.
\end{proof}

\begin{figure}[htbp]
 \centering
 \includegraphics[width=0.8\textwidth]{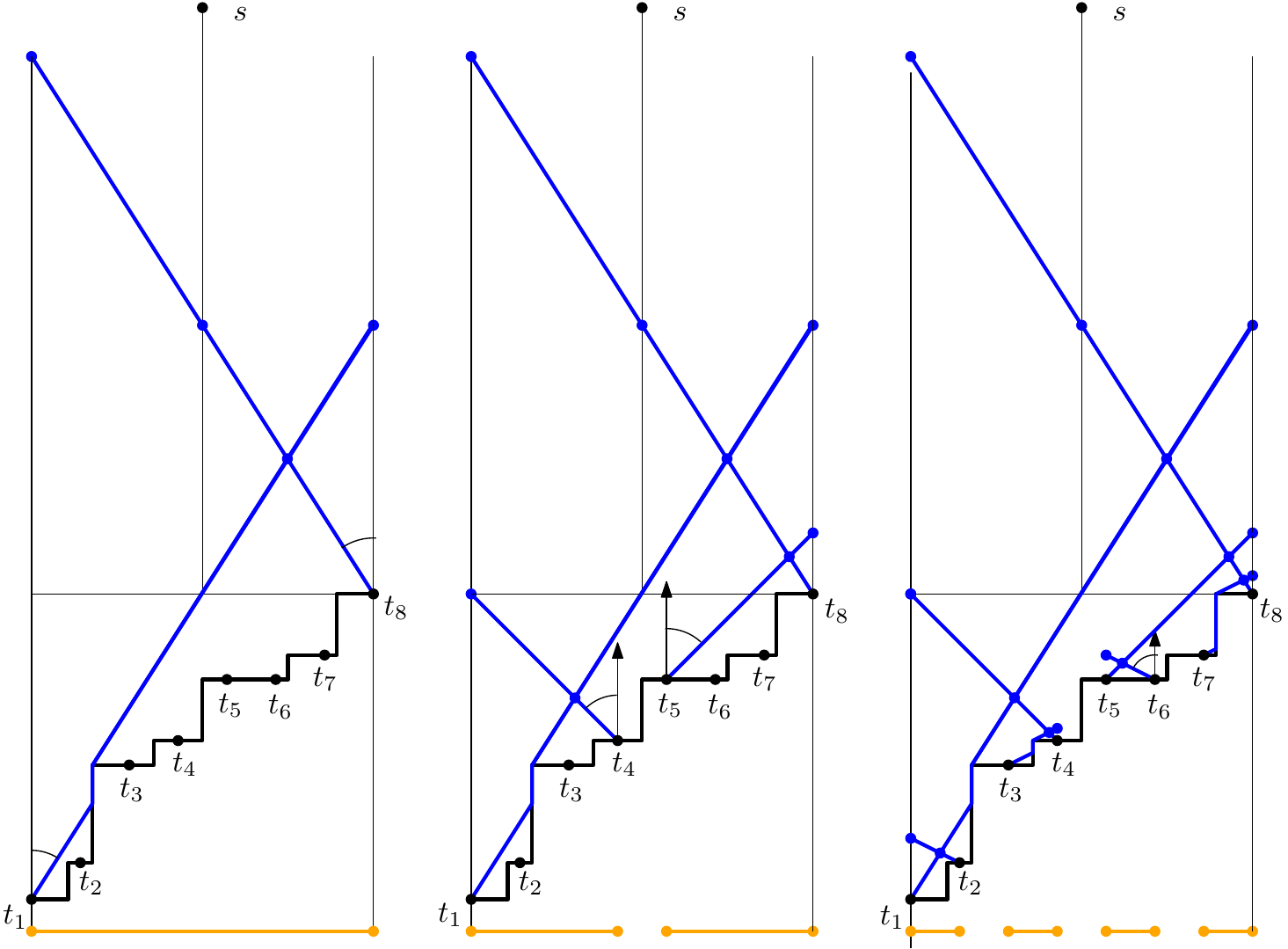}
 \caption{The paths $\gamma_q$ added to graph $G$ at level $j=0,1,2$ for $m=2^3=8$ points.
 The intervals $[t_a,t_b]$ at level $j$ are indicated below the staircase path $L$.}
    \label{fig:SLT2}
\end{figure}

\begin{lemma}\label{lem:stairs}
Let $0<\eps<1$, let $s=(0,\eps^{-1/2})$ be a point on the $y$-axis, and let $L$ be an $x$- and $y$-monotone increasing staircase path between the vertical lines $x=\pm \,\frac12$, such that the right endpoint of $L$ is $(\frac12,0)$ on the $x$-axis. Then there exists a geometric graph $G$ comprised of $L$ and additional edges of weight $O(\eps^{-1/2})$ such that $G$ contains, for every $t\in L$, an $st$-path $P_{st}$ with $\|P_{st}\|\leq (1+O(\eps))\,\|st\|$.
\end{lemma}
We can adjust the construction above as follows; refer to Fig.~\ref{fig:SLT2}.
\begin{proof}
Assume w.l.o.g.\ that $\eps=2^{-k}$ for some $k\in \mathbb{N}$.
Let $T=\{t_i: i=1,\ldots , 2^{k+1}\}$ be $2^{k+1}$ points in $L$ on equally spaced vertical lines, with spacing $1/(2^{k+1}-1)< \eps$.
Consider the standard binary partition of  $\{1,\ldots , 2^{k+1}\}$ into intervals as in the previous proof.

For every $q\in \{1,\ldots , 2^{k+1}\}$, we define a polygonal path $\gamma_q$ with one endpoint at $t_q$; see Fig.~\ref{fig:SLT2}. Let $j\geq 0$ be the smallest level such that $t_q$ is an endpoint of some interval
$I_q$ at level $j$. If $t_q$ is the right endpoint of $I_q$, then let $\gamma_q$
be the line segment of direction $\frac{\pi}{2}+2^{(j-k)/2}$ such that its
$x$-projection is $T_q$. If $t_q$ is the left endpoint of $I_q$, then $\gamma_q$ will be an $x$- and $y$-monotone path whose $x$-projection is $T_q$, and its edges will be vertical segments along $L$ and segments of direction $\alpha_q=\frac{\pi}{2}-2^{(j-k)/2}$. Specifically, $\gamma_q$ starts from $t_q$ with a line of direction $\alpha_q$. Whenever $\gamma_q$ encounters a vertical edge of $L$, it follows it upward until its upper endpoint, and then continues in direction $\alpha_q$.

Let $G$ be the union of all paths $\gamma_q$ for $q=1,\ldots , 2^{k+1}$, as well as the path $L$, and the vertical segment from $s$ to the origin. This completes the construction of $G$.

\smallskip\noindent\emph{Lightness analysis.}
We show that $\|G\|=\|L\|+O(\eps^{-1/2})$. The distance between $s$ and $L$ is $\eps^{-1/2}$. For every $q\in \{1,\ldots, 2^{k+1}\}$, the path $\gamma_q$ is composed of vertical segments along $L$, and nonvertical segments whose total weight is the same as $\|\ell_q\|$ in the proof of Lemma~\ref{lem:shallow}, where we have seen that
$\sum_{q=1}^{2^{k+1}} \|\ell_q\| = O(\eps^{-1/2})$. Consequently, $\|G\|=\|L\|+O(\eps^{-1/2})$.

\smallskip\noindent\emph{Source-stretch analysis.}
We show that $G$ contains an $st_q$-path of weight $(1+(\eps))\|st_q\|$ for all $q=1,\ldots, 2^{k+1}$. Denoting $y(t_q)$ the $y$-coordinate of point $t_q$, we have $\|st_q\|\geq \eps^{-1/2}+|y(t_q)|$.
For each interval $[t_a,t_b]$ in the binary tree, the paths $\gamma_a$ and $\gamma_b$ cross above the portion of $L$ between $t_a$ and $t_b$. Consequently, for every point $t_q$, the union of the $k+1$ paths $\gamma$ corresponding to the intervals that contain $t_q$ must contain a $y$-monotonically increasing path $P_q$ from $t_q$ to $s$.
The $y$-projection of this path has length $\eps^{-1/2}+|y(t_q)|$. Some of the edges of this path are vertical.
Consider the union of all nonvertical edges $e$ of $P_q$ along a path $\gamma$ at level $j$,
which all have direction $\frac{\pi}{2}\pm 2^{(j-k)/2}$. The difference
between the length of $e$ and the $y$-projection of $e$ is bounded by the same analysis as in the proof of Lemma~\ref{lem:shallow}. Summation over all levels yields $O(\eps^{1/2})\leq \|st_q\|\cdot O(\eps)$.

Finally, for an arbitrary point $t\in L$, we have $\|st\|\geq |y(s)-y(t)|=\eps^{-1/2}+|y(t)|$, and $G$ contains an $st$-path that consists of an $st_q$-path from $s$ to the closest point $t_q$ to the right of $t$, followed by an $x$- and $y$-monotone path along $L$ in which the total length of the horizontal edges is bounded by $1/2^k\leq \eps$ (and the length of vertical segment might be arbitrary). We use the lower bound $\|st\|\geq \eps^{-1/2}+|y(t)|$. The vertical segments between $t_q$ and $t$ do not contribute to the error term $\|st\|-(\eps^{-1/2}+|y(t)|)$. The analysis in the proof of Lemma~\ref{lem:shallow}
yields $\|st\|-(\eps^{-1/2}+|y(t)|)\leq O(\sqrt{\eps})\leq O(\eps)\|st\|$.
\end{proof}

Note that the source-stretch analysis assumed that the vertical edges of an $st$-path (along the vertical edges of $L$)
do not accumulate any error. Consequently, the same analysis carries over if we replace the vertical edges of $L$ by $(\frac{\pi}{2}\pm \frac{\sqrt{\eps}}{2})$-angle-bounded paths.
The key observation is that in the proof of Lemma~\ref{lem:stairs}, all nonvertical edges have directions that differ from vertical (i.e., from $\frac{\pi}{2}$) by $\sqrt{\eps}$ or more.

\begin{corollary}\label{cor:stairs}
Let $0<\eps<1$, let $s=(0,\eps^{-1/2})$ be a point on the $y$-axis, and let $L$ be a path between the vertical lines $x=\pm \, \frac12$, obtained from an $x$- and $y$-monotone increasing staircase path with the right endpoint at $(\frac12,0)$ on the $x$-axis,
by replacing the vertical edges with  $y$-monotonically increasing
$(\frac{\pi}{2}\pm \frac{\sqrt{\eps}}{2})$-angle-bounded paths. Then there exists a geometric graph $G$ that contains $L$ and additional edges of weight $O(\eps^{-1/2})$ such that $G$ contains, for every $t\in L$, an $st$-path $P_{st}$ with $\|P_{st}\|\leq (1+O(\eps))\,\|st\|$.
\end{corollary}

\subsection{Combination of Shallow-Light Trees}
\label{ssec:combination}

The combination of two \textsf{SLT}s yields a light $(1+\eps)$-spanner between points on two staircases.

\begin{lemma}\label{lem:combine2}
Let $R$ be an axis-parallel rectangle of width $1$ and height $2\eps^{-1/2}$; and let $L_1$ (resp., $L_2$) be a staircase path from the lower-left (upper-left) corner of $R$ to a point on the vertical line passing through the right side of $R$, lying below (above) $R$; see Fig.~\ref{fig:combinations}. Then there exists a geometric graph comprised of $L_1\cup L_2$ and additional edges of weight $O(\eps^{-1/2})$ that contains an $ab$-path $P_{ab}$ with $\|P_{ab}\|\leq (1+O(\eps))\,\|ab\|$ for any $a\in L_1$ and any $b\in L_2$.
\end{lemma}

\begin{figure}[htbp]
 \centering
 \includegraphics[width=0.65\textwidth]{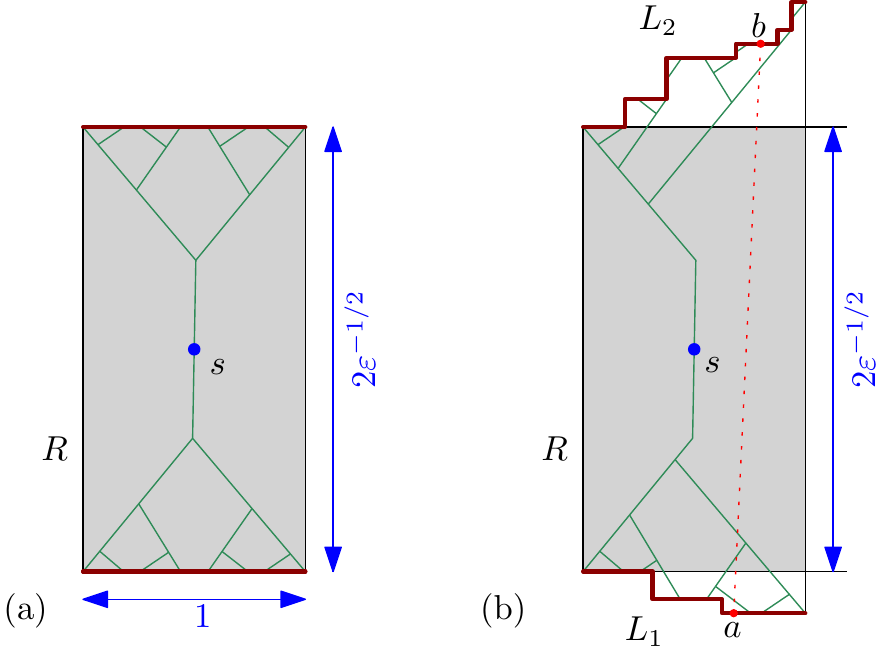}
 \caption{(a) A combination of two \textsf{SLT}s between the two horizontal sides of $R$.
 (b) A combination of two \textsf{SLT}s between two staircases above and below $R$, respectively.}
    \label{fig:combinations}
\end{figure}

\begin{proof}
Let $s$ be the center of the rectangle $R$. Let $G$ be the geometric graph formed by the \textsf{SLT}s from the source $s$ to $L_1$ and $L_2$, resp., using Lemma~\ref{lem:stairs}. By construction, $\|G\|=\|L_1\|+\|L_2\|+O(\eps^{-1/2})$.
It remains to show that $G$ has the desired spanning ratio.
Let $a\in L_1$ and $b\in L_2$. Let $h_a$ be the distance of $a$ from bottom side of $R$, and $h_b$ the distance of $b$ from the top side of $R$. By Lemma~\ref{lem:shallow}, the two \textsf{SLT}s jointly contain an $ab$-path $P_{ab}$ of length
$\|P_{ab}\|\leq (1+O(\eps))\,(\|as\|+\|bs\|)$.

On the one hand, $s$ is the center of $R$, and so $\|as\|+\|bs\|\leq \diam(R)+h_a+h_b\leq (1+\frac{\eps}{8})2{\eps}^{-1/2}+h_a+h_b$.
On the other hand, $\|ab\|\geq \mathrm{height}(R)+h_a+h_b=2{\eps}^{-1/2}+h_a+h_b$.
Overall, $\|P_{ab}\|\leq (1+O(\eps))(1+\frac{\eps}{8})\,\|ab\| \leq (1+O(\eps))\|ab\|$.
\end{proof}

\section{Reduction to Directional Spanners in Histograms}
\label{sec:redux}

In this section, we present our general strategy for the proof of Theorem~\ref{thm:UB},
and reduce the construction of a light $(1+\eps)$-spanner for a point set $S$ in the plane to a special case of \emph{directional} spanners for
a point set on the boundary of faces in a (modified) \emph{window partition}.

\subparagraph{Directional $(1+\eps)$-Spanners.}
Our strategy to construct a $(1+\eps)$-spanner for a point set $S$ is to partition the interval of directions $[0,\pi)$ into $O(\eps^{-1/2})$ intervals, each of length $O(\eps^{1/2})$. For each interval $D\subset [0,\pi)$, we construct a geometric graph that serves point pairs $a,b\in S$ with $\mathrm{dir}(ab)\in D$. Then the union of these graphs over all $O(\eps^{-1/2})$ intervals will serve all point pairs $ab\in S$. The following definition formalizes this idea.

\begin{definition}
Let $S$ be a finite point set in $\mathbb{R}^2$, and let $D\subset [0,\pi)$ be a set of directions. A geometric graph $G$ is a \emph{directional $(1+\eps)$-spanner} for $S$ and $D$ if $G$ contains an $ab$-path of weight at most $(1+\eps)\|ab\|$ for every $a,b\in S$ with $\mathrm{dir}(ab)\in D$.
\end{definition}

In Section~\ref{sec:hist}, we modify the standard window partition algorithm and partition a bounding box of $S$ into \emph{fuzzy staircases} and \emph{tame histograms} (defined below). We also construct directional spanners for point pairs $a,b\in S$, where $ab$ is a chord of a face in this partition. A line segment $ab$ is a \emph{chord} of a simple polygon $P$ if $a,b\in \partial P$, and $ab\subset P$.
The \emph{perimeter} of a simple polygon $P$, denoted $\per(P)$, is the total weight of the edges of $P$; and the \emph{horizontal perimeter}, denoted $\hper(P)$, is the total weight of the horizontal edges of $P$.

\begin{restatable}{lemma}{tilinglemma}\label{lem:tiling}
We can subdivide a simple rectilinear polygon $P$ into a collection $\mathcal{F}$ of fuzzy staircases and tame histograms of total perimeter
$\sum_{F\in \mathcal{F}} \per(F) \leq O(\eps^{-1/2}\per(P))$
and total horizontal perimeter
$\sum_{F\in \mathcal{F}} \hper(F) \leq O(\per(P))$.
\end{restatable}

\begin{restatable}{lemma}{histlemma}\label{lem:hist}
Let $F$ be a fuzzy staircase or a tame histogram, $S\subset \partial F$ a finite point set, $\eps>0$, and $D=[\frac{\pi-\sqrt{\eps}}{2},\frac{\pi+\sqrt{\eps}}{2}]$ an interval of nearly vertical directions. Then there exists a geometric graph of weight $O(\per(F)+\eps^{-1/2}\, \hper(F))$ such that for all $a,b\in S$, if $ab$ is a chord of $F$ and $\mathrm{dir}(ab)\in D$, then $G$ contains an $ab$-path of weight at most $(1+O(\eps))\|ab\|$.
\end{restatable}

We prove Lemmas~\ref{lem:tiling} and~\ref{lem:hist} in Section~\ref{sec:hist}.
In  the remainder of this section, we show that these lemmas imply Lemma~\ref{lem:dir}, which in turn implies Theorem~\ref{thm:UB}.

\begin{lemma}\label{lem:dir}
Let $S\subset \mathbb{R}^2$ be a finite point set, $\eps>0$, and $D\subset [0,\pi)$ an interval of length $\sqrt{\eps}$. Then there exists a \emph{directional $(1+\eps)$-spanner} for $S$ and $D$ of weight $O(\eps^{-1/2}\, \|\MST\|)$.
\end{lemma}
\begin{proof}

We may assume, by applying a suitable rotation, that $D=[\frac{\pi-\sqrt{\eps}}{2},\frac{\pi+\sqrt{\eps}}{2}]$, that is, an interval of nearly vertical directions. We construct a directional $(1+\eps)$-spanner for $S$ and $D$ of weight $O(\eps^{-1/2}\cdot \|\MST(S)\|)$.

Assume w.l.o.g.\ that the unit square $U=[0,1]^2$ is the minimum axis-parallel bounding square of $S$. In particular, $S$ has two points on two opposite sides of $U$, and so $1\leq \diam(S)\leq \|\MST(S)\|$. Our initial graph $G_0$ is composed of the boundary of $U$ and a \emph{rectilinear MST} of $S$, where $\|G_0\|=O(\|\MST(S)\|)$.
Since each edge of $G_0$ is on the boundary of at most two faces, the total perimeter of all faces of $G_0$ is also $O(\|\MST(S)\|)$.
Lemma~\ref{lem:tiling} yields subdivisions of the faces of $G_0$ into a collection $\mathcal{F}$ of fuzzy staircases and tame histograms of total perimeter $\sum_{F\in \mathcal{F}}\per(F)=O(\eps^{-1/2}\|MST(S)\|)$
and horizontal perimeter $\sum_{F\in \mathcal{F}}\hper(F)=O(\|MST(S)\|)$,

Let $K(S)$ be the complete graph induced by $S$. For each face $F\in \mathcal{F}$, let $S_F$ be the set of all intersection points between the boundary $\partial F$ and the edges of $K(S)$. For each face $F$, Lemma~\ref{lem:hist} yields a geometric graph $G_F$ of weight $O(\per(F)+\eps^{-1/2}\hper(F))$ with respect to the finite point set $S_F\subset \partial F$.

We can now put the pieces back together.
Let $G$ be the union of $G_0$ and the graphs $G_F$ for all $F\in \mathcal{F}$.
The weight of $G$ is bounded by
$\|G\|=\|G_0\|+\sum_{F\in \mathcal{F}} \|G_F\|
= O(\|\MST(S)\|+\sum_{F\in \mathcal{F}} (\per(F)+\eps^{-1/2}\hper(F)))
= O(\eps^{-1/2}\|\MST(S)\|)$.

Let $a,b\in S$. The edges of $G_0$ subdivide the line segment $ab$ into a path $v_0v_1\ldots  v_m$ of collinear segments, each of which is a chord of some face in $\mathcal{F}$. For $i=1,\ldots , m$, graph $G$ contains a $v_{i-1}v_i$-path of weight at most $(1+\eps)\|v_{i-1}v_i\|$. The concatenation of these paths
is an $ab$-path of length at most $\sum_{i=1}^m (1+\eps)\|v_{i-1}v_i\| =(1+\eps)\|ab\|$, as required.
\end{proof}

\begin{proof}[Proof of Theorem~\ref{thm:UB}]
Let $S$ be a finite set of points in the plane. Let $\eps>0$ be given.
For $k=\lceil \pi\eps^{-1/2}\rceil$, we partition the space of directions as $[0,\pi)=\bigcup_{i=1}^{k} D_i$, into $k$ intervals of equal length.
By Lemma~\ref{lem:dir}, there exists a directional $(1+\eps)$-spanner of weight $O(\eps^{-1/2}\|\MST(S)\|)$ for $S$ and $D_i$ for all $i$.
Let $G=\bigcup_{i=1}^{k} G_i$. For every point pair $s,t\in S$,
we have $\mathrm{dir}(st)\in D_i$ for some $i\in \{1,\ldots,k\}$, and $G_i\subset G$ contains an $st$-path of weight at most $(1+\eps)\|st\|$. Consequently, $G$ is a Euclidean Steiner $(1+\eps)$-spanner for $S$.
The weight of $G$ is
$\|G\|\leq \sum_{i=1}^{k} \|G_i\|
\leq {\lceil\pi\eps^{-1/2}\rceil}\cdot O(\eps^{-1/2}\|\MST(S)\|)
\leq O(\eps^{-1}\|\MST(S)\|)$, as required.
\end{proof}

\section{Construction of Directional Spanners for Staircases}
\label{sec:staircases}

In this section, we handle the special case where the points are on a $x$- and $y$-monotone rectilinear path $L$, which is called a \emph{staircase path}.
Our recursive construction uses a type of polygons that we define now.
A \emph{shadow polygon} is bounded by a staircase path $L$ and a single line segment of
slope $\eps^{-1/2}$; see Fig.~\ref{fig:shadow}(a) for examples.

\begin{figure}[htbp]
 \centering
 \includegraphics[width=0.9\textwidth]{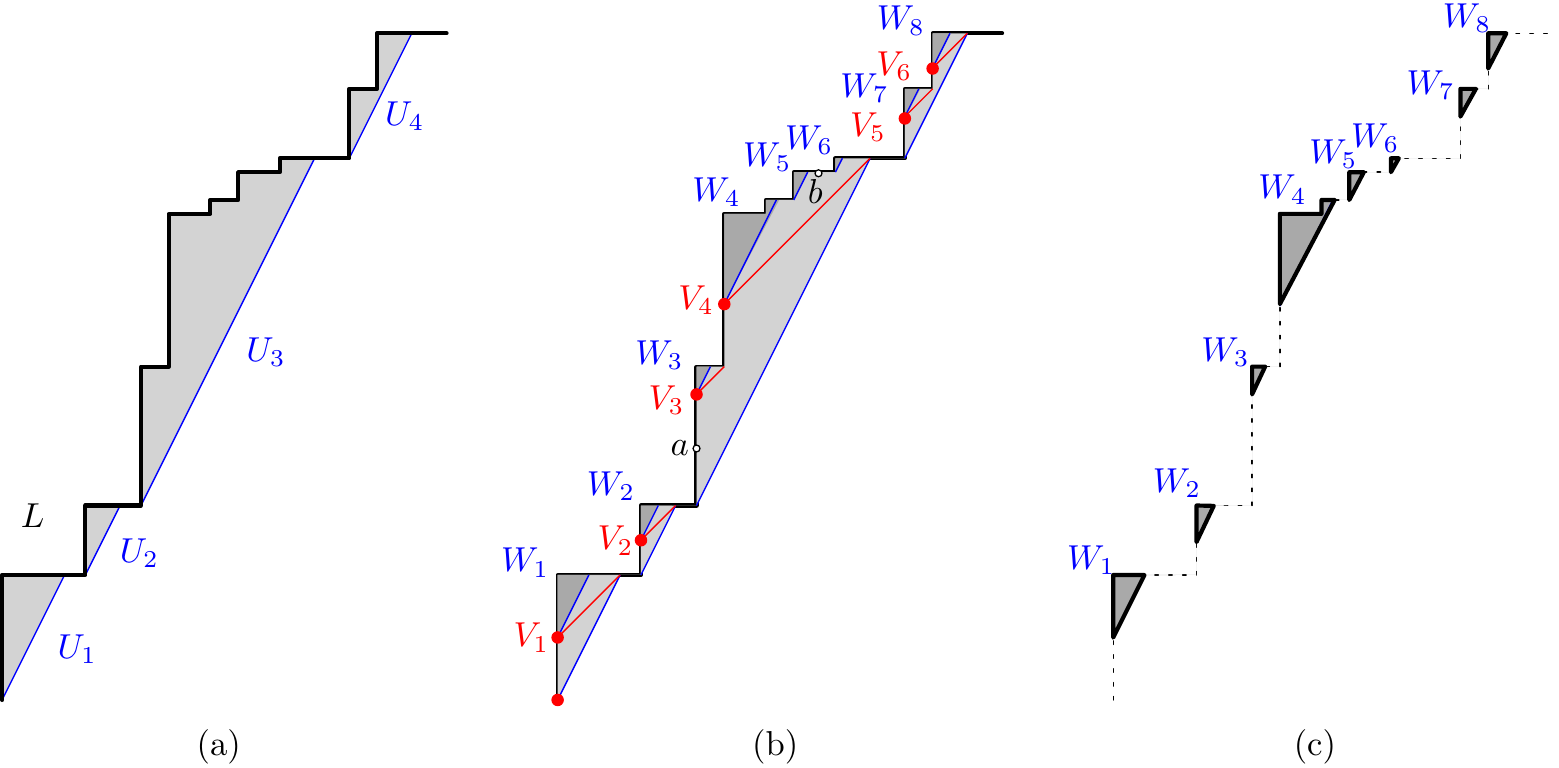}
 \caption{(a) A staircase path $L$; the shadow of vertical edges of $L$ is shaded light gray.
 (b) The shadow of the horizontal edges in the subpolygons is shaded dark gray.
 (c) Recursive subproblems generated in the proof of Lemma~\ref{lem:staircase}.}
    \label{fig:shadow}
\end{figure}

\begin{lemma}\label{lem:staircase}
Let $L$ be an $x$- and $y$-monotonically increasing staircase path, and let $S\subset L$ be a finite point set.
Then there exists a geometric graph $G$ comprised of $L$ and additional edges of weight $O(\eps^{-1/2}\mathrm{width}(L))$ such that $G$ contains a path $P_{ab}$ of weight $\|P_{ab}\|\leq (1+O(\eps))\|ab\|$ for any $a,b\in L$ where $\slope(ab)\geq \eps^{-1/2}$ and the line segment $ab$ lies below $L$.
\end{lemma}
\begin{proof}
If $a,b\in L$ and $ab$ lies below $L$, then either both $a$ and $b$ are in the same edge of $L$ (hence $L$ contains a straight-line path $ab$), or one point in $\{a, b\}$ is on a vertical edge of $L$ and the other is on a horizontal edge of $L$. We may assume w.l.o.g.\ that $a$ is on a vertical edge and $b$ is on a horizontal edge of $L$.

Let $A$ be the set of all points $p$ such that there exists $a\in L$ on some vertical edge of $L$ such that $\mathrm{slope}(ap)\geq \eps^{-1/2}$ and $ap$ is below $L$; see Fig.~\ref{fig:shadow}(a).
The set $A$ is not necessarily connected, the connected components of $A$ are shadow polygons for disjoint subpaths of $L$. Let $\mathcal{U}$ be the set of these shadow polygons. Note that for every pair $a,b\in L$, if $\mathrm{slope}(ab)\geq \eps^{-1/2}$ and $ab$ lies below the path $L$, then $ab$ lies in some polygon in $\mathcal{U}$.
For each polygon $U\in \mathcal{U}$, we construct a geometric graph $G(U)$ of weight $O(\eps^{-1/2}\mathrm{width}(U))$
such that $G(U)\cup L$ is a directional spanner for the point pairs in $S\cap U$. Then $L$ together with $\bigcup_{U\in \mathcal{U}} G(U)$ is a directional spanner for all possible $ab$ pairs. Since the shadow polygons in $\mathcal{U}$
are adjacent to disjoint portions of $L$, we have $\sum_{U\in \mathcal{U}} \mathrm{width}(U)\leq \mathrm{width}(L)$, and so $\sum_{U\in \mathcal{U}}\|G(U)\|=O(\eps^{-1/2}\mathrm{width}(L))$, as required.

\smallskip\noindent\textbf{Recursive Construction.}
For all $U\in \mathcal{U}$, we construct $G(U)$ recursively as follows. Assume that $|S\cap U|\geq 2$.
Let $B(U)$ be the set of all points $p\in U$ for which there exists a point $b$ on some horizontal edge of $U$ such that $bp\subset U$ and $\mathrm{slope}(ab)\geq \frac12 \eps^{-1/2}$; see Fig.~\ref{fig:shadow}(b).
The set $B(U)$ may be disconnected, each component is a simple polygon bounded by a contiguous portion of $L$ and a line segment of slope $\frac12 \eps^{-1/2}$. Denote by $\mathcal{V}$ the set of connected components of $B(U)$.

For every $V\in \mathcal{V}$, let $C(V)$ be the set of all points $p\in V$ for which there exists a point $a$ on some vertical edge of $V$ such that $ap\subset V$ and $\mathrm{slope}(ap)\geq \eps^{-1/2}$; see Fig.~\ref{fig:shadow}(b). Again, the set $C(V)$ may be disconnected, each component is a shadow polygon.
Denote by $\mathcal{W}$ the set of all connected components of $C(V)$ for all $V\in \mathcal{V}$.

Since $\hght(W)/\wdth(W)=\eps^{-1/2}$ for all $W\in \mathcal{W}$ and
$\hght(V)/\wdth(V)=\frac12\eps^{-1/2}$ for all $V\in \mathcal{V}$, we have
\begin{align}
\sum_{W\in \mathcal{W}}\mathrm{width}(W)
&= \sqrt{\eps}\cdot \sum_{W\in \mathcal{W}}\mathrm{height}(W)
\leq \sqrt{\eps}\cdot \sum_{V\in \mathcal{V}}\mathrm{height}(V)\nonumber\\
&= \frac12\, \sum_{V\in \mathcal{V}}\mathrm{width}(V)
\leq \frac12\, \sum_{U\in \mathcal{U}}\mathrm{width}(U). \label{eq:width0}
\end{align}

For every polygon $V\in \mathcal{V}$, let $s_V$ be the bottom vertex of $V$.
We construct a sequence of shallow-light trees from source $s_V$ as follows.
For every nonnegative integer $i\geq 0$, let $h_i$ be a horizontal line at distance $\mathrm{height}(V)/2^i$ above $s_V$. If there is any point in $S$ between $h_i$ and $h_{i+1}$,
then we construct an \textsf{SLT} from $s_V$ to the portion of $L$ between $h_i$ and $h_{i+1}$.
By Lemma~\ref{lem:staircase}, the total weight of these trees is $O(\eps^{-1/2}\mathrm{width}(V))$.
Over all $V\in \mathcal{V}$, the weight of these \textsf{SLT}s is $\sum_{V\in \mathcal{V}} O(\eps^{-1/2}\mathrm{width}(V))
=O(\eps^{-1/2}\mathrm{width}(U))$.
For all $V\in \mathcal{V}$, we also add the boundary $\partial V$ to our spanner,
at a cost of $\sum_{V\in \mathcal{V}} \per(V) =\sum_{V\in \mathcal{V}} O(\eps^{-1/2}\mathrm{width}(V))=O(\eps^{-1/2}\mathrm{width}(U))$.
This completes the description of one iteration.
Recurse on all $W\in \mathcal{W}$ that contain any point in $S$.

\smallskip\noindent\emph{Lightness analysis.}
Each iteration of the algorithm, for a shadow polygon $U$, constructs \textsf{SLT}s of total weight $O(\eps^{-1/2}\mathrm{width}(U))$, and produces subproblems whose combined width is at most $\frac12\wdth(U)$ by Equation~\eqref{eq:width0}.
Consequently, summation over all levels of the recursion yields
$\|G(U)\|=O(\eps^{-1/2}\mathrm{width}(U) \cdot \sum_{i\geq 0}2^{-i})=O(\eps^{-1/2}\mathrm{width}(U))$, as required.

\smallskip\noindent\emph{Stretch analysis.}
Now consider a point pair $a,b\in S$ such that $\mathrm{slope}(ab)\geq \eps^{-1/2}$,
$a$ is in a vertical edge of $L$, and $b$ is in a horizontal edge of $L$. Assume that $U$ is the smallest shadow polygon in the recursive algorithm above that contains both $a$ and $b$. Then $b\in V$ for some $V\in \mathcal{V}$, and $a$ is at or below vertex $s_V$ of $V$. Now we can find an $ab$-path $P_{ab}$ as follows:
First construct a $y$-monotonically increasing path from $a$ to $V_S$ along vertical edges of $L$ and along edges of some polygons in $\mathcal{V}$; all these edges have slope larger than $\frac{1}{2}\eps^{-1/2}$. Then from $s_V$, follow an \textsf{SLT} to $b$.
All edges of $P_{ab}$ from $a$ to $s_V$ have slope at least $\frac12\eps^{-1/2}$, and so their directions differ from vertical by at most $\mathrm{arctan}(2\eps^{1/2})\leq 3\eps^{1/2}$, using the Taylor expansion of $\tan(x)$ near $0$.
By Lemma~\ref{lem:angle2} the stretch factor of the paths from $a$ to $s_V$ and the path $as_Vb$ are each at most $1+O(\eps)$.
By Lemma~\ref{lem:staircase}, the \textsf{SLT} contains a path from $s_V$ to $b$ with stretch factor $1+O(\eps)$. Overall, $\|P_{ab}\|\leq (1+O(\eps))\|ab\|$.
\end{proof}

In Section~\ref{ssec:6.4}, we show that Lemma~\ref{lem:staircase} continues to hold if we replace the vertical edges of the staircase $L$ with angle-bounded paths. Furthermore, the horizontal edges can also be replaced by $x$-monotone paths of approximately the same length. Specifically, Lemma~\ref{lem:tame-staircase} below generalizes this result to so-called tame paths (defined in Section~\ref{sec:hist}).


\section{Construction of Directional Spanners in Histograms}
\label{sec:hist}

We would like to partition a simple rectilinear polygon $P$ into a collection $\mathcal{F}$ of simple polygons (faces), and then design a directional $(1+\eps)$-spanner for each face $F\in \mathcal{F}$ such that the total weight of these spanners is under control. Lemma~\ref{lem:staircase} tells us that we can handle staircase polygons efficiently.
The standard window partition~\cite{Link00,Suri90} would partition $P$ into histograms as indicated in Fig.~\ref{fig:window}(a).

\begin{figure}[htbp]
 \centering
 \includegraphics[width=0.95\textwidth]{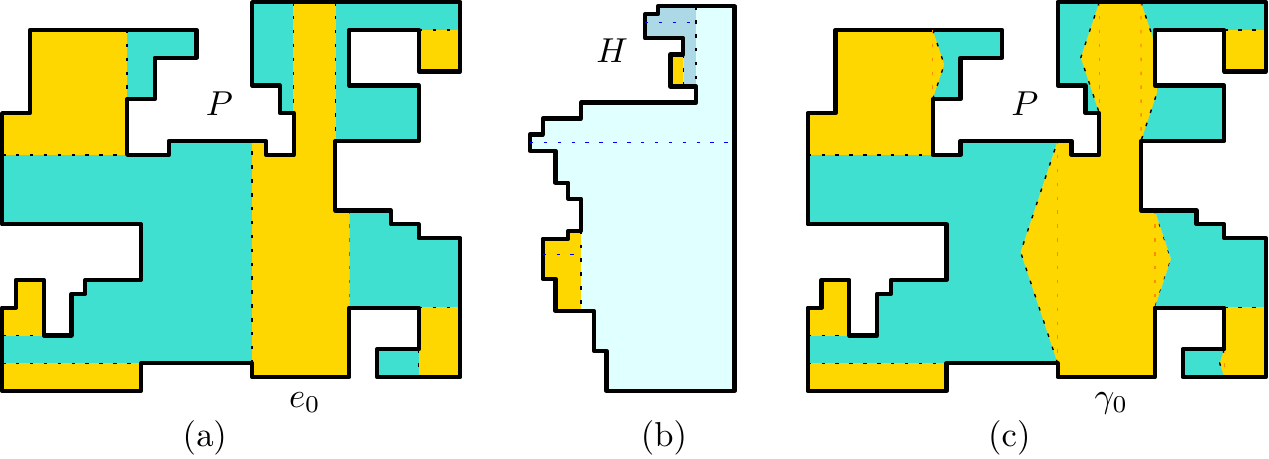}
 \caption{(a) A standard window partition of a rectilinear polygon $P$ into histograms, starting from a horizontal edge $e_0$. (b) A decomposition of a $y$-monotone histogram into staircase polygons.
 (c) The modified window partition of a rectilinear polygon $P$ into $x$-monotone $\Lambda$-histograms and $y$-monotone fuzzy histograms.}
    \label{fig:window}
\end{figure}

However, the worst-case weight of a standard decomposition of a histogram $H$ into staircases is  $\Theta(\per(H)\log n)$, where $n$ is the number of vertices of $H$. We cannot afford a $\log n$ factor (or any function of the cardinality $|S|$). To overcome this technical difficulty, we replace the vertical edges by nearly vertical $\delta$-angle-bounded paths (cf.~Lemma~\ref{lem:angle2}). By setting $\delta=\Theta(\eps^{-1/2})$, such paths provide enough flexibility to keep the weight of the subdivision under control; and our results on \textsf{SLT}s carry over to such ``modified'' staircases with only a constant-factor increase in the total weight.

We introduce some terminology; refer to Fig.~\ref{fig:Delta}. Let $\Lambda\geq 8$ be a constant.
\begin{itemize}\itemsep 0pt
\item A \emph{$\Lambda$-path} is a $y$-monotone path in which every edge is vertical, or has slope $\pm\,\Lambda\eps^{-1/2}$.
\item A \emph{$\Lambda$-histogram} is a simple polygon obtained from a histogram by replacing vertical edges with some $\Lambda$-paths.
    A $\Lambda$-histogram is \emph{$x$-monotone} (resp., \emph{$y$-monotone}) if it is obtained from an $x$-monotone (resp., $y$-monotone) histogram.
\item A \emph{fuzzy staircase} is a simple polygon bounded by a path $pqr$, where $pq$ is horizontal and $\slope(qr)=\pm\,\Lambda\eps^{-1/2}$, and a $pr$-path obtained from an $x$- and $y$-monotone staircase by replacing vertical edges with some $\Lambda$-paths.
\item A \emph{fuzzy histogram} is a simple polygon bounded by a $y$-monotone rectilinear path $L$ and a path
$\gamma$ of one or two edges of slopes $\pm\,\Lambda\eps^{-1/2}$; if the latter path has two edges, then its interior vertex is a reflex vertex of the polygon.
\end{itemize}

\begin{figure}[htbp]
 \centering
 \includegraphics[width=0.75\textwidth]{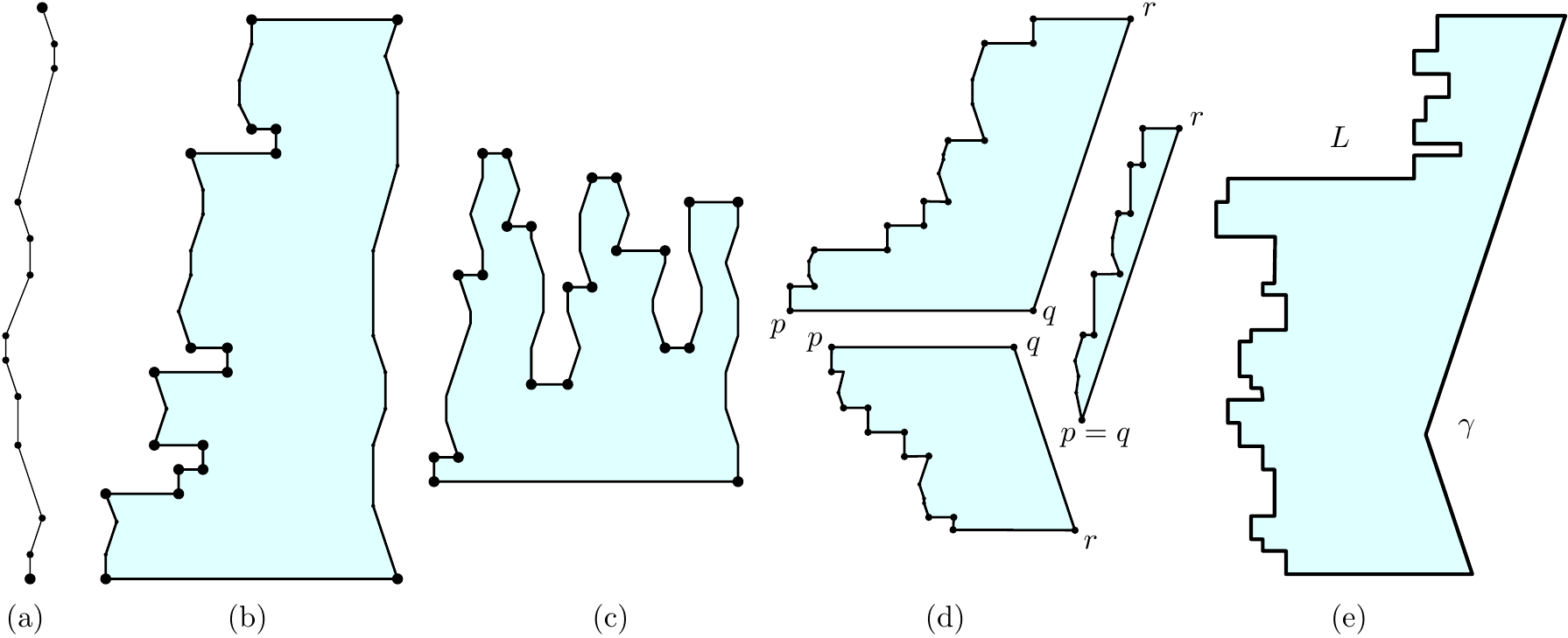}
 \caption{(a) A $\Lambda$-path.
 (b) A $y$-monotone $\Lambda$-histogram.
 (c) An $x$-monotone $\Lambda$-histogram.
 (d) Three fuzzy staircases. (e) A fuzzy histogram}
    \label{fig:Delta}
\end{figure}

\subsection{Fuzzy Window Decomposition}
\label{ssec:6.1}

Let $R$ be a rectilinear simple polygon. By modifying the standard window-partition,
we show how to partition $R$ into $y$-monotone fuzzy histograms and $x$-monotone $\Lambda$-histograms.

\begin{lemma}\label{lem:window}
Every rectilinear simple polygon $P$ can be subdivided into a collection $\mathcal{H}$ of
$x$-monotone $\Lambda$-histograms and $y$-monotone fuzzy histograms such that
$\sum_{H\in \mathcal{H}}\per(H)=O(\per(P))$.
Furthermore, in the $x$-monotone $\Lambda$-histograms, there is no chord
between interior points of two ascending (resp., two descending) $\Lambda$-paths.
\end{lemma}
\begin{proof}
We describe a recursion on instances $(R,\gamma)$ of the following two types:
\begin{itemize}\itemsep 0pt
\item $R$ is a rectilinear simple polygon in which $\gamma$ is a horizontal side;
\item $R$ is a fuzzy histogram in which $\gamma$ is the path of one or two edges of slopes $\pm\,\Lambda\eps^{-1/2}$;
       if $\gamma$ has two edges, then its interior vertex is a reflex vertex of $R$.
\end{itemize}
Initially, let $R=P$ and $\gamma$ an arbitrary horizontal edge of $P$.
In one iteration, consider an instance $(R,\gamma)$; see Fig.~\ref{fig:window}(c) for examples.

First let $(R,\gamma)$ be an instance, where $\gamma$ is a horizontal edge of $R$.
We define a histogram $H_0\subset R$ as the set of all points that can be connected to a point in $\gamma$ by a vertical line segment in $R$. Let $\mathcal{C}_0$ be the collection of connected components $C_0$ of $R\setminus H_0$.
Each component $C_0\in \mathcal{C}_0$ is a simple rectilinear polygon that has a unique vertical edge (\emph{window}) $w$ along the boundary of $H_0$. Let $H(C_0)$ denote the set of points in $P$ that can be connected to some point in $w$ by line segments of
slope both $\Lambda\eps^{-1/2}$ and $-\Lambda\eps^{-1/2}$.
Now we can define $H=H_0\cup \left(\bigcup_{C_0\in \mathcal{C}_0} H(C_0)\right)$, and note that $H$ is an $x$-monotone $\Lambda$-histogram. Let $\mathcal{C}$ be the collection of the connected components of $R\setminus H$. Note that each $C\in \mathcal{C}$ is a simple polygon in which all edges are axis-parallel except for the common boundary $\gamma(C)$ with $H$, which consists of one or two line segments of slopes $\pm\,\Lambda\eps^{-1/2}$. Furthermore, if $\gamma$ has two edges, then its interior vertex is a reflex vertex of $C$. Consequently, we can recurse on the instances
$(C,\gamma(C))$ for all $C\in \mathcal{C}$.

Next let $(R,\gamma)$ be an instance, where $\gamma$ is a path of one or two edges of slopes $\pm\, \Lambda\eps^{-1/2}$. We define a fuzzy histogram $H\subset R$ as the set of all points that can be connected to a point in $\gamma$ by a horizontal line segment in $H$.
Let $\mathcal{C}$ be the collection of the connected components of $R\setminus H$. Each component $C\in \mathcal{C}$ is a rectilinear simple polygon that has a unique horizontal edge $e(C)$ along the boundary of $H$. We then recurse on the instances $(C,e(C))$ for all $C\in \mathcal{C}$.

\smallskip\noindent\emph{Weight analysis.}
Consider the input rectilinear polygon $P$, and the subdivision created by the above algorithm. Each iteration of the recursion handles an instance $(R,\gamma)$, and creates a $x$-monotone $\Lambda$-histogram or a $y$-monotone fuzzy histogram $H$, which is not partitioned further. The cost of creating $H$ equals to the weight of the subdivision curves on boundary between $H$ and $R\setminus H$. Each component of $\partial H\cap \partial(R\setminus H)$ is the curve $\gamma'$ of a subproblem $(R',\gamma')$ at the next level of the recursion. Since the connected components of $R\setminus H$ are disjoint, the regions $R'$ at the next level are also disjoint.

At the next level of the recursion, in each instance $(R',\gamma')$, we charge the weight of $\gamma'$ to the common boundary of $R'$ and the original polygon $P$ as follows. When $\gamma'$ is a horizontal side of $R'$, then we charge the vertical projection of $\gamma$ to the boundary of $R'$. The projection consists of one or more horizontal line segments of total weight $\|\gamma'\|$. The projection lies on the boundary of the $x$-monotone $\Lambda$-histogram $H'$ that is not partitioned further. When $\gamma$ is a path of one or two edges of slope
$\pm\, \Lambda\eps^{-1/2}$, we charge the horizontal projection of $\gamma$ to the boundary of $R'$. The projection consists of one or more vertical line segments of total weight at least $\|\gamma'\|/(1+\eps/2\Lambda^2)$. The projection is on the boundary of the $y$-monotone fuzzy historgam $H$' that is not partitioned further. In both cases, the portion of the boundary of $P$ that we charge is on the boundary of the histogram $H'$, and will not be charged in recursive  subproblems. Overall every line segment $ab$ on the boundary of the input polygon $P$ receives a charge of at most $\|ab\|/(1+\eps/2\Lambda^2)=O(\|ab\|)$, hence the total weight of all subdivision edges is bounded by $O(\per(P))$.
\end{proof}

\old{
The recursion on $(R,\gamma)$ has two types.
In the first type of the recursion, $\gamma$ is a horizontal side
of $R$. We consider the histogram $H_0$ defined by a subset of $R$, that is reachable from $\gamma$ by vertical line segments.
Each of these vertical line segments creates exactly one connected component in $\mathcal{C}_0$, unless the line segment belongs to the boundary of $R$ and in that case it creates none. We call these segments windows. Now, consider each component $C_0\in \mathcal{C}_0$ and its corresponding window $w$.
We replace them by $\Lambda$-paths, and these paths together with
$H_0$, give us a $x$-monotone $\Lambda$-histogram $H$.
Note, $R\setminus H$ creates a collection of simple polygons $\mathcal{C}$ made of horizontal and vertical sides except one  side that is comprised of one or two line segments of slopes $\pm\, \Lambda\eps^{-1/2}$. This brings us to the second type of the recursion, where $\gamma$ is a path of one or two line segments of slopes
$\pm\, \Lambda\eps^{-1/2}$. For each such polygon
$C\in\mathcal{C}$, we obtain a fuzzy histogram by connecting all points of the current polygon that are reachable from $\gamma$ by horizontal segments.
Clearly, the fuzzy histograms create only horizontal windows,
and do not create any vertical window. This means that any horizontal chord defined by two points $a,b\in R$, can intersect at most two vertical windows created due to the first type of the recursion; which are eventually replaced by $\Lambda$-paths. Otherwise, it will contradict the construction of $x$-monotone $\Lambda$-histogram $H$ as well as the fuzzy histograms.
From the similar arguments, it follows that every vertical chord can intersect at most two horizontal windows. Thereby, every vertical chord intersects maximum three faces of the subdivision.

Next, we argue the perimeter of the subdivision.
Due to the first type of the recursion, we create a set of vertical segments which may or may not be part of input polygon boundary. If they are part of the boundary then we are done.
But, for the ones that are not (the windows), we can charge them against the vertical boundary. One can simply visualize the boundary as a vertical line segment starting from the topmost point to the bottommost point of the input rectangle, and the windows are a set of intervals on it. A slight caveat is that we could end up charging the same part of the boundary again and again. What comes to rescue is that we know that
no horizontal chord can intersect more than two windows. This means every time there is a horizontal line that intersects more than two windows, we get a vertical boundary from the input rectangle to whom we can charge. Hence, the total sum of these intervals length is some constant times the length of the vertical boundary of the input rectangle. However, the vertical segments are replaced by $\Lambda$-paths. This does not create much trouble though, as due to Lemma~\ref{lem:angle2}, we know that that length of an $\delta$-angle bounded is at most $(1+\delta^2/2)$ times its original length.
Then, in the second type of the recursion we create the fuzzy histograms. These histograms create a set of horizontal segments (windows). By symmetric argument it is possible show that the sum of horizontal segments is at most some constant times the length of the horizontal boundary. Therefore, by summing these costs we obtain the total perimeter of the subdivision $O(per(P))$.
}

\subsection{From $y$-Monotone Histograms to Fuzzy Staircases}
\label{ssec:ymonotne}

In this section, we show how to construct a directional $(1+\eps)$-spanner for points on the boundary of an $y$-monotone histogram $H$, using the boundary of $H$ and additional edges of weight $O(\eps^{-1/2}\,\per(H))$.

\begin{figure}[htbp]
 \centering
 \includegraphics[width=0.9\textwidth]{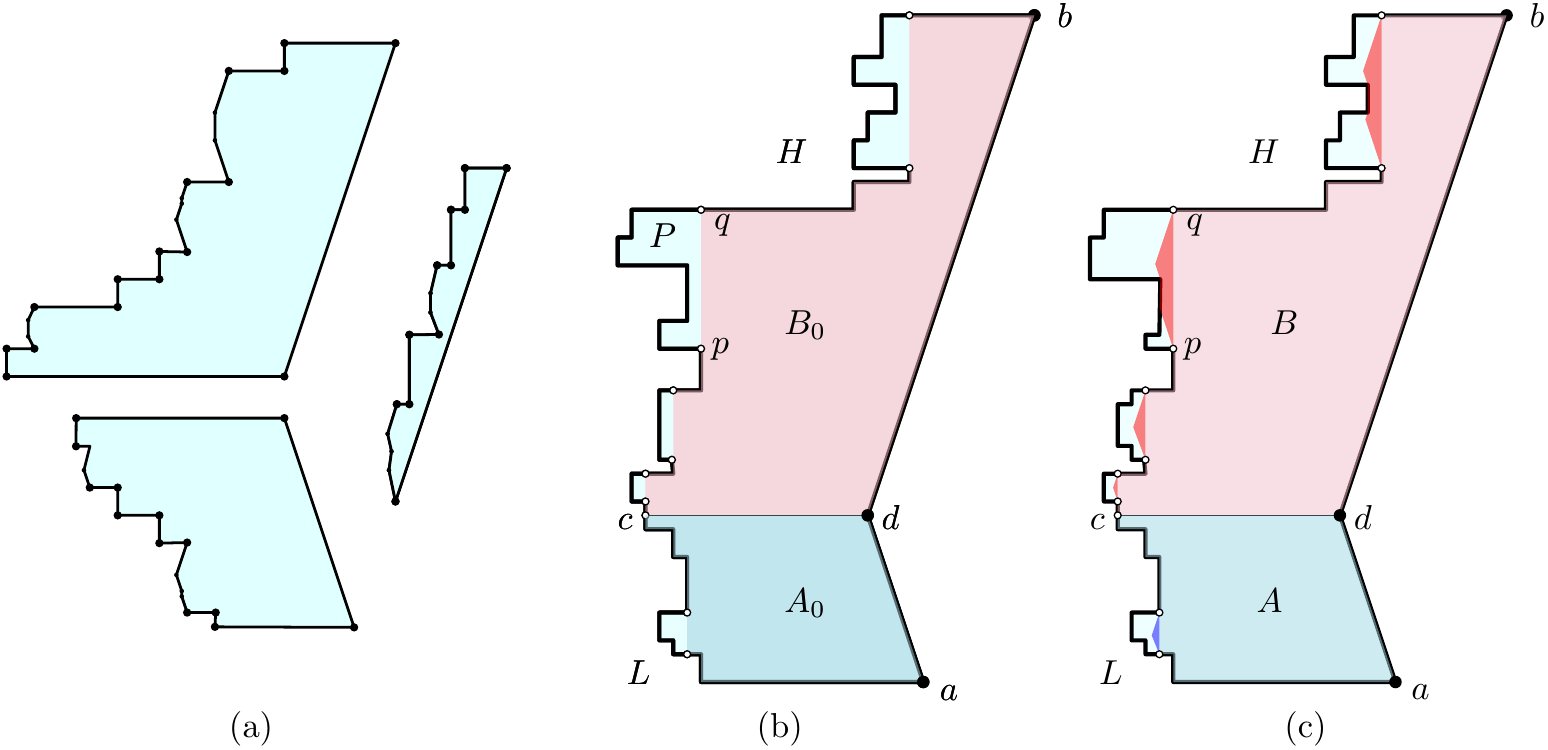}
 \caption{(a) Three $\Lambda$-staircases.
 (b) A fuzzy histogram $H$, bounded by $L$ and a path $adb$.
 (c) The subdivision of $H$ into two $\Lambda$-staircases $A$ and $B$, and fuzzy histograms.}
    \label{fig:fuzzy}
\end{figure}

We recursively decompose a fuzzy histogram into fuzzy staircases.
\begin{lemma}\label{lem:fuzzy}
We can subdivide a fuzzy histogram into a family $\mathcal{F}$ of fuzzy staircases
using subdivision edges of total weight $O(\eps^{-1/2}\,\per(H))$ such that
$\sum_{F\in \mathcal{F}} \hper(F)=O(\per(H))$.
\end{lemma}
\begin{proof}
Let $H$ be a fuzzy histogram, bounded by an $y$-monotone increasing rectilinear $ab$-path $L$,
and an $ab$-path $\gamma$ of one or two edges of slope $\pm\,\Lambda\eps^{-1/2}$; see Fig.~\ref{fig:fuzzy}(c).
We subdivide $H$ recursively. Specifically, we subdivide $H$ into
two fuzzy staircases and a collection of fuzzy histograms,
and recurse on the remaining fuzzy histograms.

\smallskip\noindent\emph{One iteration.}
Let $d$ be the leftmost vertex of $\gamma$, which subdivides $\gamma$ into subpaths
$\gamma_1$ and $\gamma_2$ above and below $d$, respectively.
Let $c\in L$ such that $y(c)=y(d)$. Note that $cd$ is a chord of $H$.
We define two initial fuzzy staircases:
Let $A_0$ be the set of points in $H$ below the line $cd$ that
can connected to some point in $cd\cup da$ by a vertical line in $H$;
and let $B_0$ be the set of points in $H$ above the line $cd$ that
can be connected to some point in $cd\cup db$ by a vertical line in $H$.

We expand $A_0$ and $B_0$ into the connected components of $H\setminus (A_0\cup B_0)$, that we call \emph{pockets}, as follows; refer to Fig.~\ref{fig:fuzzy}(d).
Let $\mathcal{P}$ be the set of all connected components of $H\setminus (A_0\cup B_0)$,
and let $P\in \mathcal{P}$. Then $P$ is a $y$-monotone histogram bounded by
a rectilinear $y$-monotone path $L(P)\subset L$, and a vertical (window)
$w$ on the boundary of $A_0$ or $B_0$.
Assume w.l.o.g.\ that $w\subset \partial A_0$.
Let $A(P)$ be the set of points $p\in P$ such that
$P$ contains two line segments, of slopes $\Lambda\eps^{-1/2}$ and $-\Lambda\eps^{-1/2}$
between $p$ and some points in $w$.
Now let
$A=A_0\cup\left(\bigcup_{P\in \mathcal{P}} A(P)\right)$;
and $B=B_0\cup\left(\bigcup_{P\in \mathcal{P}} B(P)\right)$ is defined analogously.
This expansion step, replaces every window $w$ by a $\Lambda$-path,
consequently $A$ and $B$ are fuzzy staircases.
Furthermore, each connected component of $H\setminus (A\cup B)$ is a
fuzzy histogram. This completes the description of one iteration.

\smallskip\noindent\emph{Lightness analysis.}
Let $\mathcal{F}$ be the set of fuzzy histograms produced by the algorithm.
Every fuzzy histogram $H_i$ in the recursion is adjacent to some subpath $L_i$ of the initial $y$-monotone rectilinear path $L$. We can measure the progress of the algorithm in terms of the total weight of the horizontal edges of $L_i$.
In one iteration, we add a horizontal subdivision edge $cd$, and subdivision paths of slopes $\pm\,\Lambda\eps^{-1/2}$ on the boundary between $A\cup B$ and $H\setminus (A\cup B)$. The vertical projection of $cd$, $ad$, and $db$ to $L$ are on the common boundary of $L$ with the fuzzy staircases $A$ or $B$. We can charge $\|cd\|$ and $\frac{\sqrt{\eps}}{\Lambda}\, (\|ad\|+\|db\|)$ to these portions of $L$. Similarly, the weight of all edges of slope $\pm\,\Lambda\eps^{-1/2}$ on the left boundary of $A$ and $B$ are charged to the horizontal edges of $L$ in the next level of the recursion. The weight of all subdivision edges is at most $\Lambda\eps^{-1/2}$ times the weight of the horizontal edges of $L$, which is $\hper(H)$. Overall, the total weight of all subdivision edges is $O(\eps^{-1/2}\,\per(H))$.
Since we directly charge the new horizontal edges (of the subdivision) to horizontal edges of $L$, we have $\sum_{F\in \mathcal{F}} \hper(F)=O(\per(H))$.
\end{proof}

\subsection{From $x$-Monotone $\Lambda$-Histograms to Tame Histograms}
\label{ssec:xmonotne}

Let $H$ be a $x$-monotone $\Lambda$-histogram produced by our modified window partition (Lemma~\ref{lem:window}). We first use horizontal lines to subdivide $H$ into smaller pieces that are easier to handle. We begin by defining the type of polygons that we want to obtain.
\begin{itemize}
\item A \emph{tame histogram} is a simple polygon $H$ bounded by a horizontal line segment $pq$ and an $pq$-path $L$ that consists of ascending or descending $\Lambda$-paths and $x$-monotone increasing horizontal edges
    with the following properties:
    (i) there is no chord between interior points of any two ascending (resp., two descending) $\Lambda$-paths; and
    (ii) for every horizontal chord $ab$, with $a,b\in L$, the subpath $L_{ab}$ of $L$ between $a$ and $b$
    satisfies $\|L_{ab}\|\leq 2\|ab\|$.
\item A \emph{tame path} is a subpath of the $pq$-path $L$ of a tame histogram.
\end{itemize}

\subparagraph{Remark.}
Importantly, when we replace a vertical line segment by a $\Lambda$-path (as in the definition of $\Lambda$-histograms), then both endpoints of such a $\Lambda$-path have the same $x$-coordinate. This property is no longer required for the $\Lambda$-paths in tame histogram. Properties (i) and (ii) in the definition of tame histogram help prevent pathological polygons, where an alternating sequence of $y$-monotonically increasing and decreasing $\Lambda$-path could accumulate unbounded perimeter. As we shall see, we construct tame histograms from $x$-monotone $\Lambda$-histograms, and properties (i) and (ii) will be easy to establish.

\subparagraph{Subdivision of $x$-Monotone $\Lambda$-Histograms.}
We partition an $x$-monotone $\Lambda$-histogram, obtained in Lemma~\ref{lem:window}, into tame histograms. This can be done by an easy sweepline algorithm. Dumitrescu and T\'oth~\cite{DumitrescuT09} used similar methods to partition an (ordinary) histogram into histograms of constant geometric dilation.

\begin{lemma}\label{lem:tame}
Let $H$ be an $x$-monotone $\Lambda$-histogram that does not have any chords between interior points of any two ascending (resp., two descending) $\Lambda$-paths.
Then we can subdivide $H$ into a collection $\mathcal{T}$ of
tame histograms such that $\sum_{T\in \mathcal{T}}\per(T)=O(\per(H))$.
\end{lemma}
\begin{proof}
Let $H$ be a $x$-monotone $\Lambda$-histogram bounded by a horizontal line segment and a path $L$ obtained from a $x$-monotone rectilinear path by replacing the vertical edges with $\Lambda$-paths. Assume that $H$ does not have any chords between interior points of any two ascending (resp., two descending) $\Lambda$-paths.

We describe a sweepline algorithm that recursively subdivides $H$ with horizontal lines; see Fig.~\ref{fig:tame}(a). Initially, set $\mathcal{T}=\emptyset$. Sweep $H$ top-down with a horizontal line $\ell$, and incrementally update $L$, $H$, and $\mathcal{T}$ as follows. Whenever, the sweepline $\ell$ contains a chord $ab$ of $H$ such that the subpath $L_{ab}$ of $L$ between $a$ and $b$ has weight $2\|ab\|$, then we add the simple polygon bounded by $ab$ and $L_{ab}$ into $\mathcal{T}$, and replace $L_{ab}$ with the line segment $ab$ in both $L$ and $H$. When the sweepline $\ell$ reaches the base of $H$, we add $H$ to $\mathcal{T}$, and return $\mathcal{T}$.

First note that each polygon added into $\mathcal{T}$ is a tame histogram. By construction, $\sum_{T\in \mathcal{T}}\per(T)$ is proportional to the sum of $\per(H)$ and the total weight of all horizontal chords $ab$ inserted by the algorithm.
At the time when we create a tame histogram bounded by $ab$ and $L_{ab}$, we can charge the weight $\|ab\|$ of the chord $ab$ to the nonhorizontal edges of the path $L_{ab}$. Since the algorithm inserts only horizontal edges, all nonhorizontal edges along this path lie on the boundary of the input polygon. Furthermore, since the nonhorizontal edges are part of $\Lambda$-paths of the input polygon, they are vertical or have slope $\pm\,\Lambda \eps^{-1/2}$, where $\Lambda\geq 8$. Consequently, the total weight of nonhorizontal edges of $L_{ab}$ is more than $\frac13 \|ab\|$. Overall, the total weight of the edges inserted by the algorithm is $O(\per(H))$, as required.
\end{proof}

The combination of Lemmas~\ref{lem:window}, \ref{lem:fuzzy}, and~\ref{lem:tame} imply Lemma~\ref{lem:tiling}.

\tilinglemma*

\begin{proof}[Proof of Lemma~\ref{lem:tiling}.]
Let $P$ be a rectilinear simple polygon. By Lemma~\ref{lem:window}, we can partition $P$ into a collection
$\mathcal{F}$ of $x$-monotone $\Lambda$-histograms and $y$-monotone fuzzy histograms such that
$\sum_{F\in \mathcal{F}}\per(F)=O(\per(P))$.

Then by Lemma~\ref{lem:fuzzy}, we can partition each $y$-monotone fuzzy histogram $F\in \mathcal{F}$ into
a collection $\mathcal{T}(F)$ of fuzzy staircases of total perimeter $O(\eps^{-1/2}\per(F))$, and total horizontal perimeter $O(\per(F))$.
By Lemma~\ref{lem:tame}, we can partition each $x$-monotone $\Lambda$-histogram $H\in \mathcal{F}$
into a collection $\mathcal{T}(F)$ of tame histograms of total perimeter $O(\per(H))$.
This implies that their horizontal perimeter is also bounded by $O(\per(H))$.

Overall, we obtain a subdivision of $R$ into a collection $\bigcup_{F\in \mathcal{F}} \mathcal{T}(F)$ of fuzzy staircases and tame histograms of total perimeter $O(\eps^{-1/2}\per(P))$ and total horizontal perimeter $O(\per(P))$, as required.
\end{proof}

\begin{figure}[htbp]
 \centering
 \includegraphics[width=0.65\textwidth]{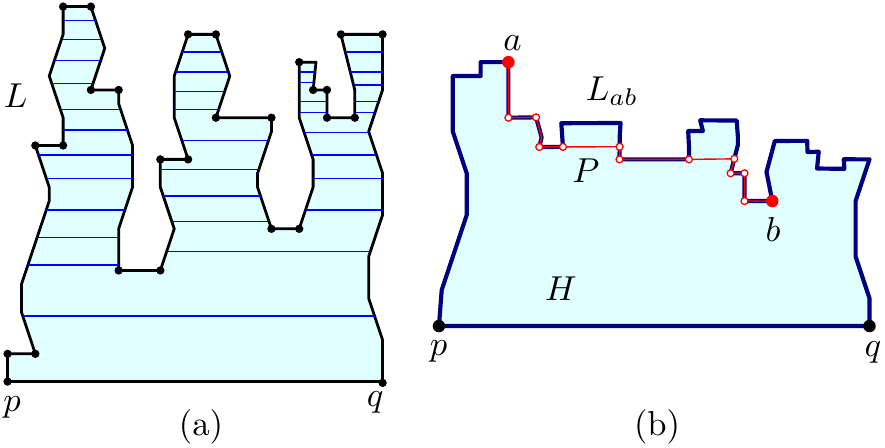}
 \caption{(a) A subdivision of an $x$-monotone $\Lambda$-histogram into tame histograms.
 (b) The $ab$-path $P$ constructed in the proof of Lemma~\ref{lem:fat}.}
    \label{fig:tame}
\end{figure}

\subsection{Directional Spanners for Tame Histograms}
\label{ssec:6.4}

In this section we prove Lemma~\ref{lem:hist} for tame histograms.
Given a tame histogram $T$ and a finite set of points $S\subset \partial T$, we construct a directional spanner for $S$ with respect to the interval $D=[\frac{\pi-\sqrt{\eps}}{2},\frac{\pi+\sqrt{\eps}}{2}]$ of nearly-vertical directions. In the discussion below, we typically use $\slope(ab)$, rather than $\dir(ab)$. Note that whenever $a,b\in S$ and $\dir(ab)\in D$, then $|\slope(ab)|\geq \eps^{-1/2}$, due to the Taylor estimate $\tan(x)\geq x+x^3/3$ for $x=\dir(ab)-\frac{\pi}{2}$ near 0.

The next lemma (Lemma~\ref{lem:fat}) establishes a key property of tame histograms: the weight of a subpath between two points in the $pq$-path can be bounded in terms of the $L_1$-distance between the two endpoints and an error term dominated by the distance between their $x$-coordinates.

\begin{lemma}\label{lem:fat}
Let $H$ be a tame histogram bounded by a horizontal segment $pq$ and a $pq$-path $L$.
Let $a,b\in L$ such that $b$ is the bottom-most point in $L_{ab}$.
Then $\|L_{ab}\|\leq 2|x(a)-x(b)|+(1+2\,\sqrt{\eps})|y(a)-y(b)|$.
\end{lemma}
\begin{proof}
Assume w.l.o.g.\ that $x(a)\leq x(b)$ and $y(a)\geq y(b)$; refer to Fig.~\ref{fig:tame}(b). We construct an $ab$-path $P_{ab}$ in $H$ that consists of portions of $L_{ab}$ and horizontal chords of $H$: Initially, we set $P$ to be the one-vertex path $P=(a)$, and then incrementally append new edges to the endpoint $c$ of $P$ until it reaches $b$. Initially, the endpoint of $P$ is $c=a$. If $c$ is in a $\Lambda$-path $\gamma$ along $L_{ab}$, but not the bottom endpoint of $\gamma$, then we extend $P$ to the bottom endpoint of $\gamma$. Else if $c$ is in a horizontal edge $e$ of $L_{ab}$, but not the right endpoint of $e$, then we extend $P$ to the right endpoint of $e$. Else $c$ is the bottom point of a $\Lambda$-path $\gamma$ and the right endpoint of a horizontal edge $e$ along $L_{ab}$, and then we extend $P$ with a horizontal chord $cd$. Such a chord exists since $b$ is the bottommost point in $L_{ab}$. The algorithm terminates with $c=b$, since in each iteration either $y(c)$ decreases, or $y(c)$ does not change but $x(c)$ increases.

Since $P_{ab}$ is a $\Lambda$-staircase path from $a$ to $b$, the total weight of
its nonhorizontal edges is at most $(1+\eps)|y(a)-y(b)|$; and the total weight of its horizontal edges is at most $|x(a)-x(b)|+\sqrt{\eps}\,|y(a)-y(b)|$.
If we replace every horizontal chords $cd$ along $P$ with the corresponding subpath $L_{cd}$ of $L$, the resulting path is precisely $L_{ab}$. Since $H$ is a tame histogram, each chord $cd$ is replaced by a path of length at most $2\|cd\|$.
Since these chords are disjoint horizontal line segments along $P_{ab}$,
the weight increase is bounded by
$\|L_{ab}\|-\|P_{ab}\|\leq |x(a)-x(b)|+\sqrt{\eps}\,|y(a)-y(b)|$.
Consequently,
\begin{align*}
\|L_{ab}\|
&=(\|L_{ab}\|-\|P_{ab}\|)+\|P_{ab}\| \\
&\leq 2|x(a)-x(b)|+(1+\sqrt{\eps}+\eps)|y(a)-y(b)|\\
&\leq 2|x(a)-x(b)|+(1+2\,\sqrt{\eps})|y(a)-y(b)|,
\end{align*}
as claimed.
\end{proof}

In Lemma~\ref{lem:combine5} and~\ref{lem:combine6} below,
we use \textsf{SLT}s to construct directional $(1+\eps)$-spanners in a tame histogram
(i) between the base and a portion of the path $L$ within a square; and
(ii) between a source $s$ and a portion of the path $L$ from $p$ to $q$.

\begin{figure}[htbp]
 \centering
 \includegraphics[width=0.7\textwidth]{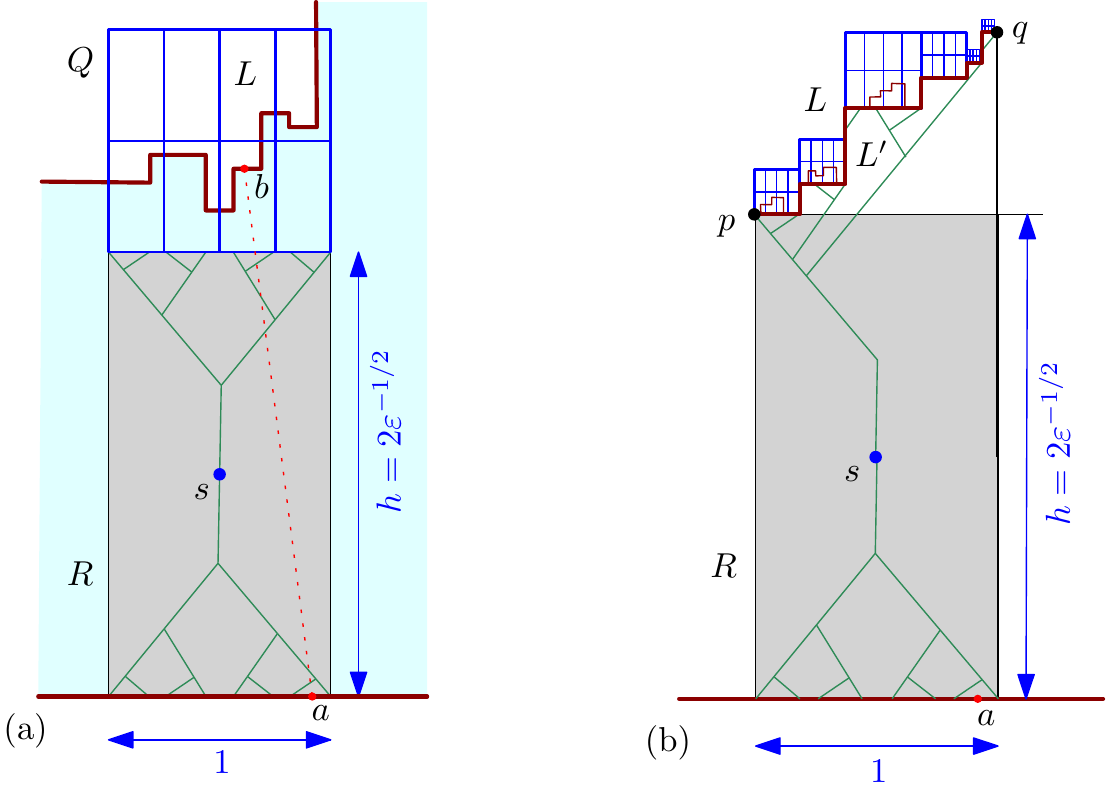}
 \caption{(a) The boundary of a tame histogram in a square $Q$ above rectangle $R$.
 (b) An adaptation of an \textsf{SLT} to tame histograms.}
    \label{fig:combinations2}
\end{figure}

\begin{lemma}\label{lem:combine5}
Let $R$ be an axis-parallel rectangle of width $1$ and height $2\eps^{-1/2}$.
Let $Q$ be a unit square adjacent to the top side of $R$, and let $L$ be a tame path in $Q$; see Fig.~\ref{fig:combinations2}(a).
Then there exists a graph $G$ comprised of $L$ and additional edges of weight $O(\eps^{-1/2})$ that contains
an $ab$-path $P_{ab}$ with $\|P_{ab}\|\leq (1+\eps)\,\|ab\|$ for any $a\in L$ and any point $b$ in the bottom side of $R$.
\end{lemma}
\begin{proof}
Let $s$ be the center of the rectangle $R$. We construct a geometric graph $G$ as follows. Let $G$ contain the bottom side of $R$, the path $L$, and the two \textsf{SLT}s from $s$ to the bottom and top sides of $R$, respectively. The graph $G$ should also contain a subdivision of $Q$ into rectangles of aspect ratio $2\eps^{-1/2}$; see Fig.\ref{fig:combinations2}(a). Specifically, we subdivide $Q$ into two congruent horizontal strips, and then subdivide each horizontal strip into rectangles
$r$ of $\wdth(r)=\sqrt{\eps}$ and $\hght(r)=\frac12$. Finally, in each rectangle $r\subset Q$ of this subdivision, if $r$ intersects $L$, then let $G$ contain a vertical line segment from a bottom-most point in $L\cap r$ to the bottom side of $r$.

\smallskip\noindent\emph{Lightness analysis.}
The weight of two \textsf{SLT}s is $O(\eps^{-1/2})$ by Lemma~\ref{lem:shallow}. Since $Q$ has unit width, the weight of the grid in $Q$ is $O(\eps^{-1/2})$, and the vertical edges in the rectangles in $r\subset Q$ are bounded above by the weight of the grid. The overall weight of $G$ is $\|L\|+O(\eps^{-1/2})$.

\smallskip\noindent\emph{Stretch-factor analysis.}
Let $a\in L$ and let $b$ be a point in the bottom side of $R$. We may assume that $a\in r$, for a rectangle $r\subset Q$ in the subdivision of $Q$. We construct an $ab$-path $P_{ab}$ as follows: Start from $a$, follow $L$ to a bottom-most point in $L\cap r$, and then use a vertical line segment to reach the bottom side of $r$. Then follow a shortest path within the grid in $Q$ to the top side of $R$, and finally use the two \textsf{SLT}s to reach $b$. For easy reference, we label some of intermediate vertices along $P_{ab}$: let $v_1$ be the bottom-most point in $L\cap r$, let $v_2$ be the bottom endpoint of the vertical segment in $r$, and $v_3$ the the first point where $P_{ab}$ reaches the top side of $R$. Note that the $y$-coordinates of these points monotonically decrease, that is, $y(a)\geq y(v_1)\geq y(v_2)\geq y(v_3)\geq y(b)$. Clearly, we have $\|ab\|\geq y(a)-y(b)=(y(a)-y(v_2))\geq \eps^{-1/2}$.

We now estimate the length of each portion of $P_{ab}$ between $a$, $v_1$, $v_2$, $v_3$, and $b$. By Lemma~\ref{lem:fat},
we have
\begin{align*}
\|P_{av_1}\|
&\leq 2\, |x(a)-x(v_1)| + (1+2\,\sqrt{\eps}) |y(a)-y(v_1)|\\
&\leq 2\,\wdth(r)+2\,\sqrt{\eps}\,\hght(r)+|y(a)-y(v_1)|\\
&\leq 3\,\sqrt{\eps}+(y(a)-y(v_1)).
\end{align*}
As $v_1v_2$ is a vertical line segment, then $\|P_{v_1 v_2}\| = y(v_1)-y(v_2)$.
Since the aspect ratio of the grid cell $r$ is $2\eps^{-1/2}$,
the length of the path $P_{v_2 v_3}$ is bound by $\|P_{v_2 v_3}\| \leq (1+\eps)(y(v_2)-y(v_3))$.
Lemma~\ref{lem:combine2} yields $\|P_{v_3 b}\|\leq (1+O(\eps)) (y(v_3)-y(b))$.
Putting the pieces together, we obtain
\begin{align*}
\|P_{ab}\|
&=\|P_{a v_1}\| + \|P_{v_1 v_2}\| + \|P_{v_2 v_3}\| + \|P_{v_3 b}\|\\
&\leq 3\,\sqrt{\eps}+(1+O(\eps))
    \Big((y(a)-y(v_1)) + (y(v_1)-y(v_2))+ (y(v_2)-y(v_3))+ (y(v_3)-y(b))\Big)\\
&\leq (1+O(\eps))(y(a)-y(b)) + 3\,\sqrt{\eps}\\
&\leq (1+O(\eps))\|ab\|,
\end{align*}
as required.
\end{proof}

\begin{lemma}\label{lem:combine6}
Let $R$ be an axis-parallel rectangle of width $1$, height $2\eps^{-1/2}$; and let $p$ be the upper-left corner of $R$, and let $q$ be a point above $R$ on vertical lines passing through the  right sides of $R$.
Let $L$ be a tame $pq$-path that lies above the line segment $pq$; see Fig.~\ref{fig:combinations2}(b).
Then there exists a geometric graph $G$ comprised of $L$ and additional edges of weight $O(\eps^{-1/2})$ that contains an $st$-path $P_{st}$ with $\|P_{st}\|\leq (1+O(\eps))\,\|st\|$ for any $s$ in the bottom side of $R$ and any $t\in L$.
\end{lemma}
\begin{proof}
Assume w.l.o.g.\ that $y(p)\leq y(q)$. For every maximal horizontal chord $ab$ of $L$, replace $L_{ab}$ with $ab$,
and denote by $L'$ the resulting $pq$-path. Then $L'$ is a $\Lambda$-staircase path. For each horizontal edge $e$ of $L'$, let $Q_e$ be an axis-parallel square of side length $\|e\|$ above $e$. Since $L$ is a tame path, each connected component of $L\setminus L'$ lies in a square $Q_e$ for some horizontal edge $e$ of $L'$; see Fig.~\ref{fig:combinations2}(b).

We construct a geometric graph $G$ as follows. Let $G$ contain two \textsf{SLT}s from the center of $R$ to the bottom side of $R$ and to $L'$, resp.,  described in Corollary~\ref{cor:stairs}.
It should also contain a subdivision of each square $Q_e$ into rectangles of aspect ratio $2\eps^{-1/2}$. Finally, in each rectangle $r\subset Q$ of this subdivision, if $r$ intersects $L$, then $G$ contains a vertical line segment from a bottom-most point in $L\cap r$ to the bottom side of $r$. The weight of the \textsf{SLT} is $O(\eps^{-1/2})$ by Corollary~\ref{cor:stairs}. Since the sum of the widths of all squares $Q_e$ is at most one 1, the total weight of the grids in $Q_e$ is also $O(\eps^{-1/2})$, and the vertical edges in the rectangles in $r\subset Q$ are bounded above by the weight of the grid. The overall weight of $G$ is $\|L\|+O(\eps^{-1/2})$.

Let $S$ be a point in the bottom side of $R$, and $t\in L$. If $t\in L'$, then the two \textsf{SLT}s jointly contain a path $P_{st}$ with $\|P_{st}\|\leq (1+O(\eps))\|st\|$ by Lemma~\ref{lem:combine2}. Otherwise, $t\in L\setminus L'$. Since $L$ is a tame path, $t$ lies in a square $Q_e$ for some horizontal $e$ of $L'$. We can construct a path $P_{st}$ as a path from $t$ to a point $t'\in e$ similarly to the proof of Lemma~\ref{lem:combine5}, followed by a path from $t'$ to $s$ in the \textsf{SLT}s.
\end{proof}

We use Lemma~\ref{lem:combine5} to construct a $(1+\eps)$-spanner between the base $pq$ and $pq$-path in a tame histogram.

\begin{lemma}\label{lem:dir2}
Let $H$ be a tame histogram bounded by a horizontal line $pq$ and $pq$-path $L$,
and let $S\subset \partial H$ be a finite point set.
Then there exists a geometric graph $G$ of weight $\|G\|=O(\eps^{-1/2}\,\per(P))$
such that $G$ contains a $ab$-path $P_{ab}$ with $\|P_{ab}\|\leq (1+\eps)\|ab\|$
for all $a\in S\cap L$ and $b\in pq$ such that $ab\subset H$ and $|\slope(ab)|\geq \eps^{-1/2}$.
\end{lemma}
\begin{proof}
We construct a collection $\mathcal{Q}$ of squares such that
for every square $Q\in \mathcal{Q}$ is adjacent to a rectangle $R(Q)$ as in
the setting of Lemma~\ref{lem:combine5}; and for every point pair $a,b\in S$,
with $a\in L$ and $b\in pq$, there is a square $Q\in \mathcal{Q}$ such that $a,b\in Q\cup R(Q)$. Let $G(Q)$ be the geometric graph in Lemma~\ref{lem:combine5} for all $Q\in \mathcal{Q}$, and let $G=\bigcup_{Q\in \mathcal{Q}}G(Q)$. Then $G$ has the required stretch factor.
It remains to construct the collection $\mathcal{Q}$ of squares, and
show that $\|G\|=O(\eps^{-1/2}\,\per(P))$.

\smallskip\noindent\emph{Construction of Squares.}
Refer to Fig.~\ref{fig:thames}.
Let $H$ be a tame histogram bounded by a horizontal line $pq$ and $pq$-path $L$.
We may assume w.l.o.g. that $p$ is the origin and $pq$ is on the positive $x$-axis,
and $h=\mathrm{height}(H)$.
Since $H$ is tame, $\|L\|\leq 2\|ab\|$, which implies that $\mathrm{height}(R)<\frac12\,\|pq\|$.
For every nonnegative integer $i\in \mathbb{N}$, let
\[\ell_i:y=\left(\frac{2\eps^{-1/2}}{1+2\eps^{-1/2}}\right)^i.\]
We tile the horizontal strip between two consecutive lines, $\ell_i$ and $\ell_{i+1}$,
by squares in two different ways, such that the midpoint of a square in one tiling is on the
boundary of two squares in the other tiling.

\begin{figure}[htbp]
 \centering
 \includegraphics[width=0.8\textwidth]{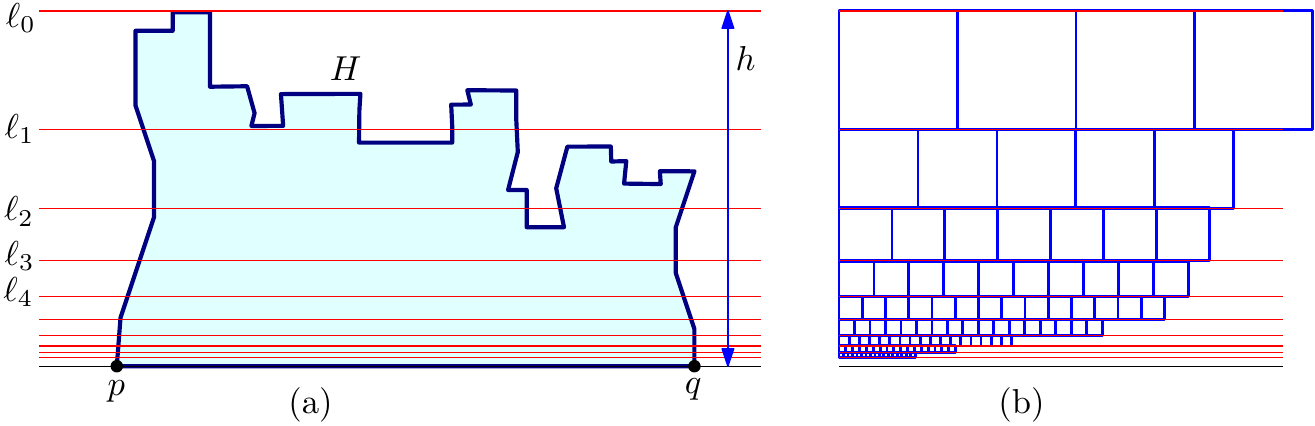}
 \caption{(a) A tame histogram $H$ and horizontal lines $\ell_i$, $i\in \mathbb{N}$ defined in the proof of Lemma~\ref{lem:dir2}.
 (b) Tiling of the horizontal strips between consecutive lines $\ell_i$ and $\ell_{i+1}$.}
    \label{fig:thames}
\end{figure}

Let $\mathcal{Q}$ be the set of squares $Q$ in the tilings defined above
such that $Q\cap S\neq \emptyset$. For each square $Q\in \mathcal{Q}$,
let $R(Q)$ be the rectangle of aspect ratio $2\eps^{-1/2}$ between $Q$ and the $x$-axis,
and $G(Q)$ the geometric graph provided by Lemma~\ref{lem:combine5}.

\smallskip\noindent\emph{Lightness Analysis.}
By Lemma~\ref{lem:combine5}, the graph $G(Q)$ comprised of $L\cap Q$ and additional edges of
length $O(\eps^{-1/2}\mathrm{width}(Q))$. For the desired bound $\|G\|\leq O(\eps^{-1/2}\per(H))$,
it is enough to prove that $\sum_{Q\in \mathcal{Q}} \mathrm{width}(Q) \leq O(\per(H))$.

We define a proximity graph $\widehat{G}$ on the squares in $\mathcal{Q}$.
The vertex set is $V(\widehat{G})=\mathcal{Q}$, and squares $Q_1,Q_2\in \mathcal{Q}$
are adjacent iff $\mathrm{dist}(Q_1,Q_2)\leq \mathrm{width}(Q_1)+\mathrm{width}(Q_2)$.
Since the squares in the horizontal strip between $\ell_i$ and $\ell_{i+1}$ form two tilings,
and the widths of the squares in adjacent horizontal strips differ by a factor close to 1,
the maximum degree in $\widehat{G}$ is $O(1)$. Consequently, $\widehat{G}$ is
$O(1)$-degenerate, and we can partition its vertex set $\mathcal{Q}$ into
$O(1)$ independent sets.

For every $Q\in \mathcal{Q}$, let $2Q$ denote the square obtained by dilating $Q$ from its center by a factor of 2.
Since $L$ contains points in $Q$, but its endpoints are outside of $2Q$, then $L$ traverses the annulus $2Q\setminus Q$ twice, which implies $\|L\cap 2Q\|\geq \mathrm{width}(Q)$. For an independent set $\mathcal{I}\subset \mathcal{Q}$, the squares $\{2Q: Q\in \mathcal{I}\}$ are pairwise disjoint.
It follows that
\begin{equation}\label{eq:ind}
\sum_{Q\in \mathcal{I}} \mathrm{width}(Q)
\leq \sum_{Q\in \mathcal{I}}\|L\cap 2Q\|
\leq \left \| L\cap \left(\bigcup_{Q\in \mathcal{I}} 2Q\right)\right\|
\leq \|L\|.\nonumber
\end{equation}
Summation over $O(1)$ independent sets yields  $\sum_{Q\in \mathcal{Q}} \mathrm{width}(Q) \leq \|L\|\leq O(\per(H))$,
as required.
\end{proof}

In the remainder of this section, we construct a directional $(1+\eps)$-spanner for points on the $x$-monotone $\Lambda$-path of a tame histogram. This is done by an adaptation of Lemma~\ref{lem:staircase}. Even though vertical edges are replaced by $\Lambda$-paths, and horizontal edges by tame paths, the weight analysis remains essentially the same.

The crucial observation in the proof Lemma~\ref{lem:staircase}
(cf.~Equation~\eqref{eq:width0})
was that if $L$ is an $x$- and $y$-monotone staircase $ab$-path,
then $\mathrm{slope}(ab)=\mathrm{height}(L)/\mathrm{width}(L)$.
We show that this equation holds approximately for a tame paths $L$,
where the width and height of $L$ are replaced by the total weight
of horizontal and nonhorizontal edges of $L$, resp., denoted $\hper(L)$ and $\vper(L)$.

\begin{lemma}\label{lem:width}
There exists a constant $\eps_0>0$ such that for all $0<\eps<\eps_0$, the following holds. If $L$ is a tame $pq$-path such that $\frac12\eps^{-1/2}\leq \slope(pq)\leq \eps^{-1/2}$,then
\begin{equation}\label{eq:width2}
\frac{3}{4} \leq \frac{\vper(L)/\hper(L)}{\slope(pq)} \leq \frac{4}{3}.
\end{equation}
\end{lemma}
\begin{proof}
Assume w.l.o.g. that $x(p)<x(q)$ and $y(p)<y(q)$; refer to Fig.~\ref{fig:shuffle}.
For $a,b\in L$, denote by $L_{ab}$ the subpath of $L$ between $a$ and $b$.
We simplify $L$ in two steps to help the weight analysis.
First, for every maximal horizontal chord $ab$ of $L$, replace $L_{ab}$ with $ab$,
and denote by $L'$ the resulting $pq$-path (Fig.~\ref{fig:shuffle}(b)).
Note that $L'$ is a $\Lambda$-staircase path comprised of horizontal edges and $\Lambda$-paths.
Second, rearrange the order of the edges in $L'$ so that all nonhorizonal edges precede all horizontal edges, and denote by $L''$ the resulting $pq$-path (Fig.~\ref{fig:shuffle}(c)).
Then $L''$ consists of a single $\Lambda$-path $\gamma$ of weight $\vper(L'')$ followed by a horizontal segment of weight $\hper(L'')$.
Since the slope and weight of the edges are unaffected by a rearrangement,
we have $\hper(L'')=\hper(L')$ and $\vper(L'')=\vper(L)$.

\begin{figure}[htbp]
 \centering
 \includegraphics[width=0.75\textwidth]{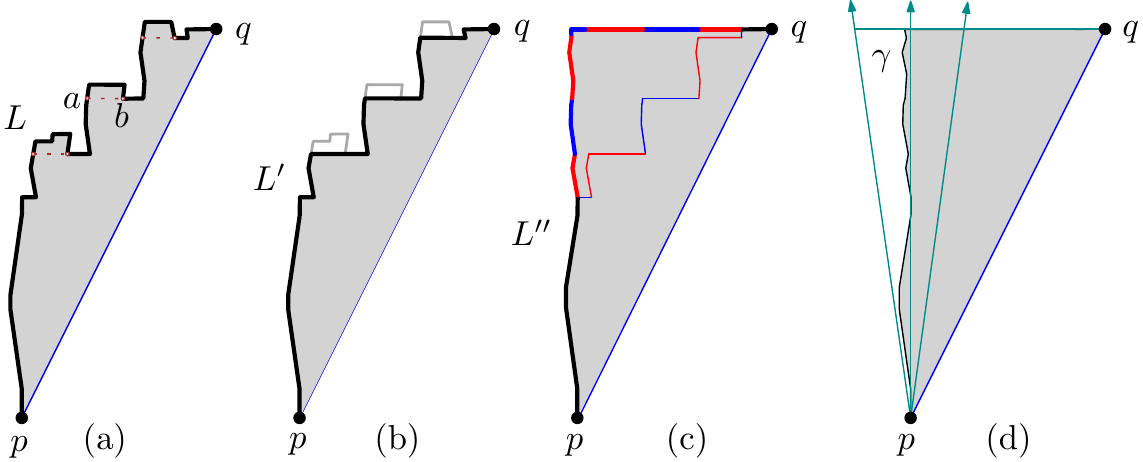}
 \caption{(a) A tame path $L$ from $p$ to $q$.
 (b) A corresponding $\Lambda$-staircase path $L'$ from $p$ to $q$.
 (c) Rearranging the edges of $L'$ produces the path $L''$.
 (d) The $\Lambda$-path of $L''$ lies in an wedge between rays of slopes $\pm\,\Lambda\eps^{-1/2}$.}
    \label{fig:shuffle}
\end{figure}

Recall that every edge of a $\Lambda$-path $\gamma$ is vertical or of slope $\pm\,\Lambda\eps^{-1/2}$. Consequently, $\gamma$ lies in a wedge with apex at $p$, bounded by rays of slopes $\pm\,\Lambda\eps^{-1/2}$ (Fig.~\ref{fig:shuffle}(d)). Then
\begin{align*}
|\hper(L')-\mathrm{width}(pq)|
&=|\hper(L'')-\mathrm{width}(pq)|
 =\mathrm{width}(\gamma)
\leq \frac{\sqrt{\eps}}{\Lambda}\, \mathrm{height}(\gamma)\\
&\leq \frac{\sqrt{\eps}}{\Lambda}\, \mathrm{height}(pq)
\leq \frac{2}{\Lambda}\,\mathrm{width}(pq)
\leq \frac15 \,\mathrm{width}(pq),
\end{align*}
and
\[
|\hper(L')-\mathrm{height}(pq)|
=|\hper(L'')-\mathrm{height}(pq)|
\leq \frac{2\sqrt{\eps}}{\Lambda}\mathrm{height}(pq)
=2\delta\, \mathrm{height}(pq),
\]
where $\delta=\sqrt{\eps}/\Lambda$. Consequently,
\begin{align}
\frac{(1-2\delta) \hght(pq)}{\frac65 \wdth(pq)} \leq
    & \frac{\vper(L')}{\hper(L')} \leq  \frac{(1+2\delta)\hght(pq)}{\frac45 \wdth(pq)}\nonumber\\
\frac56\,(1-2\delta)\,\slope(pq)\leq
    & \frac{\vper(L')}{\hper(L')} \leq  \frac54\,(1+2\delta)\,\slope(pq)\label{eq:shuffle}
\end{align}

Since $L$ is a tame path, it can be reconstructed from $L'$ by replacing some disjoint horizontal segments $ab\subset L'$ (i.e., the maximal horizontal chords of $L$) by paths of weight at most $2\|ab\|$ above $ab$. This implies that
$0\leq \vper(L)-\vper(L')\leq \hper(L')\leq 3\sqrt{\eps}$
and $|\hper(L)-\hper(L')|\leq \delta\,\hper(L')$.
Combined with \eqref{eq:shuffle}, we obtain
 \begin{align}
\frac{\vper(L')}{(1+\delta)\,\vper(L')} \leq
    & \frac{\vper(L)}{\hper(L)} \leq  \frac{(1+3\sqrt{\eps})\hper(L')}{(1-\delta)\vper(L')}\nonumber\\
\frac56\,\frac{1-2\delta}{1+\delta}\,\slope(pq)\leq
    & \frac{\vper(L')}{\hper(L')} \leq  \frac54\,\frac{(1+2\delta)(1+3\sqrt{\eps})}{1-\delta}\,\slope(pq).
\end{align}
Now \eqref{eq:width2} follows if $\eps>0$ is sufficiently small, bounded above by a suitable constant $\eps_0>0$.
\end{proof}

\begin{figure}[htbp]
 \centering
 \includegraphics[width=0.9\textwidth]{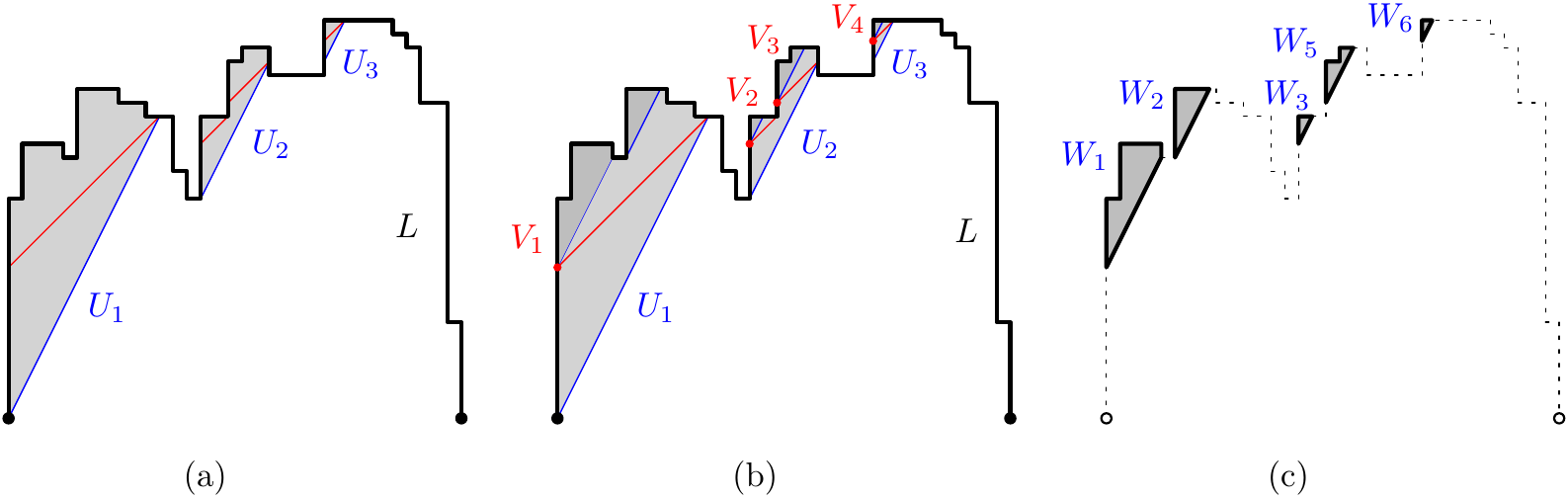}
 \caption{(a) A tame path $L$. The shadow of the ascending $\Lambda$-paths of $L$ is shaded light gray.
 (b) The shadow of the horizontal edges and descending $\Lambda$-paths is shaded dark gray.
 (c) Recursive subproblems generated in the proof of Lemma~\ref{lem:tame-staircase}.}
    \label{fig:tame-shadow}
\end{figure}

As noted above, the following lemma is an adaptation of Lemma~\ref{lem:staircase} to tame paths. Due to Lemma~\ref{lem:width}, the  recursive weight analysis carries over to this case. For clarity, we present the complete proof.

\begin{lemma}\label{lem:tame-staircase}
Let $L$ be a tame path and let $S\subset L$ be a finite point set.
Then there exists a geometric graph $G$ comprised of $L$ and additional edges of weight $O(\eps^{-1/2}\hper(L))$
such that $G$ contains a path $P_{ab}$ of weight $\|P_{ab}\|\leq (1+O(\eps))\|ab\|$ for any $a,b\in L$ such that $|\slope(ab)|\geq \eps^{-1/2}$ and the line segment $ab$ lies below $L$.
\end{lemma}
\begin{proof}
We construct $G$ as a union of two graphs, $G^+$ and $G^-$, where $G^+$ is a spanner for $\{a,b\}$ pairs with $\mathrm{slope}(ab)>0$ and $G^-$ for $\mathrm{slope}(ab)<0$. We focus on $G^+$, as the case of $G^-$ is analogous.

Let $a, b\in S$ such that $\mathrm{slope}(ab)\geq \eps^{-1/2}$ and $ab$ lies below $L$.
Without loss of generality, we may assume $y(a)<y(b)$. Since $ab$ is below $L$, point $a$ cannot be an interior of a horizontal edge of $L$. By property (i) of tame histograms, $a$ cannot be in the interior of  descending $\Lambda$-path. Consequently, $a$ is a point in an ascending $\Lambda$-path.

Let $A$ be the set of all points $p$ below $L$ such that there exists $a\in L$ on some ascending $\Lambda$-path of $L$
such that $\mathrm{slope}(ap)\geq \eps^{-1/2}$ and $ap$ is below $L$; see Fig.~\ref{fig:tame-shadow}(a).
The set $A$ is not necessarily connected, the connected components of $A$ are bounded by disjoint subpaths of $L$
and a line segment of slope $\eps^{-1/2}$. Let $\mathcal{U}$ be the set of these components.
By construction every pair $a,b\in L$ with $\mathrm{slope}(ab)\geq \eps^{-1/2}$ and $ab\subset H$ lies in a polygon in $\mathcal{U}$.
For each polygon $U\in \mathcal{U}$, we construct a geometric graph $G^+(U)$ of weight $O(\eps^{-1/2}\hper(U))$
such that $G^+(U)\cup L$ is a directional $(1+\eps)$-spanner for the points in $S\cap U$. Then $L$ together with $\bigcup_{U\in \mathcal{U}} G^+(U)$ is $(1+\eps)$-spanner for all possible $ab$ pairs. Since the polygons in $\mathcal{U}$ are adjacent to disjoint portions of $L$, we have $\sum_{U\in \mathcal{U}} \hper(U)\leq \hper(L)$,
and so $\sum_{U\in \mathcal{U}}\|G^+(U)\|=O(\eps^{-1/2}\hper(L))$, as required.

\smallskip\noindent\textbf{Recursive Construction.}
For each $U\in \mathcal{U}$, we construct $G^+(U)$ recursively as follows. Assume that $|S\cap U|\geq 2$.
Let $B(U)$ be the set of all points $p\in U$ for which there exists a point $b$ on some horizontal edge or descending $\Lambda$-path of $L\cap U$ such that $bp\subset U$ and $\mathrm{slope}(ab)\geq \frac12 \eps^{-1/2}$; see Fig.~\ref{fig:tame-shadow}(b).
The set $B(U)$ may be disconnected, each component is a simple polygon bounded by a subpath of $L$ and a line segment of slope $\frac12 \eps^{-1/2}$. Denote by $\mathcal{V}$ the set of connected components of $B(U)$.

For every $V\in \mathcal{V}$, let $C(V)$ be the set of all points $p\in V$ for which there exists a point $a$ on some vertical edge of $V$ such that $ap\subset V$ and $\mathrm{slope}(ap)\geq \eps^{-1/2}$; see Fig.~\ref{fig:tame-shadow}(b). Again, the set $C(V)$ may be disconnected, the components are simple polygons adjacent to disjoint subpaths of $L$. Denote by $\mathcal{W}$ the set of all connected components of $C(V)$ for all $V\in \mathcal{V}$.

We can apply Lemma~\ref{lem:width}  with slope $\eps^{-1/2}$ for all $W\in \mathcal{W}$;
and  with slope $\frac12\,\eps^{-1/2}$ for all $V\in \mathcal{V}$. Then
\begin{align}
\sum_{W\in \mathcal{W}}\hper(W)
&\leq \frac43\cdot \sqrt{\eps}\cdot \sum_{W\in \mathcal{W}}\vper(W)
\leq \frac43\cdot \sqrt{\eps}\cdot \sum_{V\in \mathcal{V}}\vper(V)\nonumber\\
&\leq \left(\frac43\right)^2\cdot \frac12\, \sum_{V\in \mathcal{V}}\hper(V)
\leq \frac89\, \sum_{U\in \mathcal{U}}\hper(U).\label{eq:width1}
\end{align}
Consequently, $\sum_{W\in \mathcal{W}} \|G^+(W)\|$ is proportional to
$\frac89\cdot \eps^{-1/2}\sum_{U\in \mathcal{U}}\hper(U)$.

For each polygon $V\in \mathcal{V}$, let $s_V$ be the bottom vertex of $V$, let $L(V)=L\cap V$ be the portion of $L$ on the boundary of $V$. We construct a sequence of \textsf{SLT}s from source $s_V$ as follows. For every nonnegative integer $i\geq 0$, let $h_i$ be a horizontal line at distance $\mathrm{height}(V)/2^i$ above $s_V$. By Let $L_i\subset L(V)$ be a maximum portion of $L(V)$ such that the corresponding $\Lambda$-staircase path $L'_i$ is on or below $h_i$ and strictly above $h_{i+1}$.
By Lemma~\ref{lem:combine6}, we can construct
a \textsf{SLT} from $s_V$ to $L_i$.
The total weight of these \textsf{SLT}s is $O(\eps^{-1/2}\hper(V))$.
Then the overall weight of these spanners is
$\sum_{V\in \mathcal{V}} O(\eps^{-1/2}\hper(V)) =O(\eps^{-1/2}\hper(U))$.
This completes the description of one iteration.
Recurse on all $W\in \mathcal{W}$ that contain any point in $S$.

\smallskip\noindent\emph{Lightness analysis.}
Each iteration of the algorithm, for a polygon $U$, constructs \textsf{SLT}s of total weight $O(\eps^{-1/2}\hper(U))$ by  Lemma~\ref{lem:combine6}, and produces subproblems whose combined horizontal perimeter at most $\frac89\hper(U)$ by Equation~\eqref{eq:width1}.
Consequently, summation over all levels of the recursion yields
$\|G^+(U)\|=
O(\eps^{-1/2}\hper(U) \cdot
\sum_{i\geq 0}\left(\frac89\right)^{-i})=O(\eps^{-1/2}\hper(U))$, as required.

\smallskip\noindent\emph{Stretch analysis.}
Now consider point pair $a,b\in S$ such that $\mathrm{slope}(ab)\geq \eps^{-1/2}$,
$a$ is in an ascending $\Lambda$-path of $L$, and $b$ is in a horizontal edge or a descending $\Lambda$-path of $L$.  Assume that $U$ is the smallest polygon in the recursive algorithm above that contains both $a$ and $b$. Then $b\in V$ for some $V\in \mathcal{V}$, and $a$ is at or below vertex $s_V$ of $V$. Now we can find an $ab$-path $P_{ab}$ as follows:
First construct a $y$-monotonically increasing path from $a$ to $s_V$ along $\Lambda$-paths of $L$ and along edges of some polygons in $\mathcal{V}$; all these edges have slope larger than $\frac{1}{2}\eps^{1/2}$. Then from $s_V$ to $b$, follow an \textsf{SLT} provided by Lemma~\ref{lem:combine6}. Specifically, there exists an integer $i\geq 0$ point $b$ lies on a subpath $L_i\subset L(V)$, where $L_i'$ is between the horizontal lines $h_i$ and $h_{i+1}$, and we can use the \textsf{SLT} between $s_V$ and $L_i$.

All edges of $P_{ab}$ from $a$ to $s_V$ have slope at least $\frac12\eps^{-1/2}$, and so their directions differ from vertical by at most $\mathrm{arctan}(2\eps^{1/2})\leq 3\eps^{1/2}$ from the Taylor expansion of $\tan(x)$ near $0$.
By Lemma~\ref{lem:angle2} the stretch factor of the paths from $a$ to $s_V$ and the path $as_Vb$ are each at most $1+O(\eps)$. By Lemma~\ref{lem:combine6} provides a path from $s_V$ to $b$ with stretch factor $1+O(\eps)$. Overall, $\|P_{ab}\|\leq (1+O(\eps))\|ab\|$.
\end{proof}

The combination of Lemmas~\ref{lem:dir2} and~\ref{lem:tame-staircase} provides a directional $(1+\eps)$-spanner for
all point pairs on the boundary of a tame histogram.

\begin{corollary}\label{cor:dir3}
Let $H$ be a tame histogram and  $S\subset \partial H$ a finite point set.
Then there exists a geometric graph $G$ of weight $\|G\|=O(\eps^{-1/2}\,\hper(H))$
such that $G$ contains a $ab$-path $P_{ab}$ with $\|P_{ab}\|\leq (1+O(\eps))\|ab\|$
for all $a,b\in S$ whenever $ab\subset H$ and $|\slope(ab)|\geq \eps^{-1/2}$.
\end{corollary}

\subsection{Directional Spanners for Fuzzy Staircases}
\label{ssec:DeltaStairs}

We can now construct a directional $(1+\eps)$-spanner for fuzzy staircase polygons.

\begin{lemma}\label{lem:dir1}
Let $F$ be a fuzzy staircase polygon and $S\subset \partial F$ a finite point set.
Then there exists a geometric graph $G$ of wight $\|G\|=O(\eps^{-1/2}\,\hper(F))$
such that $G$ contains a $ab$-path $P_{ab}$ with $\|P_{ab}\|\leq (1+\eps)\|ab\|$
for all $a,b\in S$ if $ab\subset F$ and $\mathrm{dir}(ab)\in D$.
\end{lemma}
\begin{proof}
Let $F$ be a fuzzy staircase polygon bounded by a horizontal segment $pq$, a segment $qr$ of slope $\Lambda\eps^{-1/2}$, and a path $L$ obtained from an $x$- and $y$-monotone staircase by replacing vertical edges with some $\Lambda$-paths.
For point pairs $a,b\in L\cap S$, Lemma~\ref{lem:tame-staircase} provides a desired spanner of weight $O(\eps^{-1/2}\hper(P))$.

It remains to construct a spanner for point pairs $a,b\in S$, where $a\in pq\cup qr$
Assume that $p$ is the origin, and $pq$ is on the positive $x$-axis.
Let $h=\hght(F) = \hght(qr)$. Since $\slope(qr)=\Lambda\eps^{-1/2}$, where $\Lambda=O(1)$,
then $h=O(\eps^{-1/2}\wdth(pq))=O(\eps^{-1/2}\hper(F))$.

\begin{figure}[htbp]
 \centering
 \includegraphics[width=0.95\textwidth]{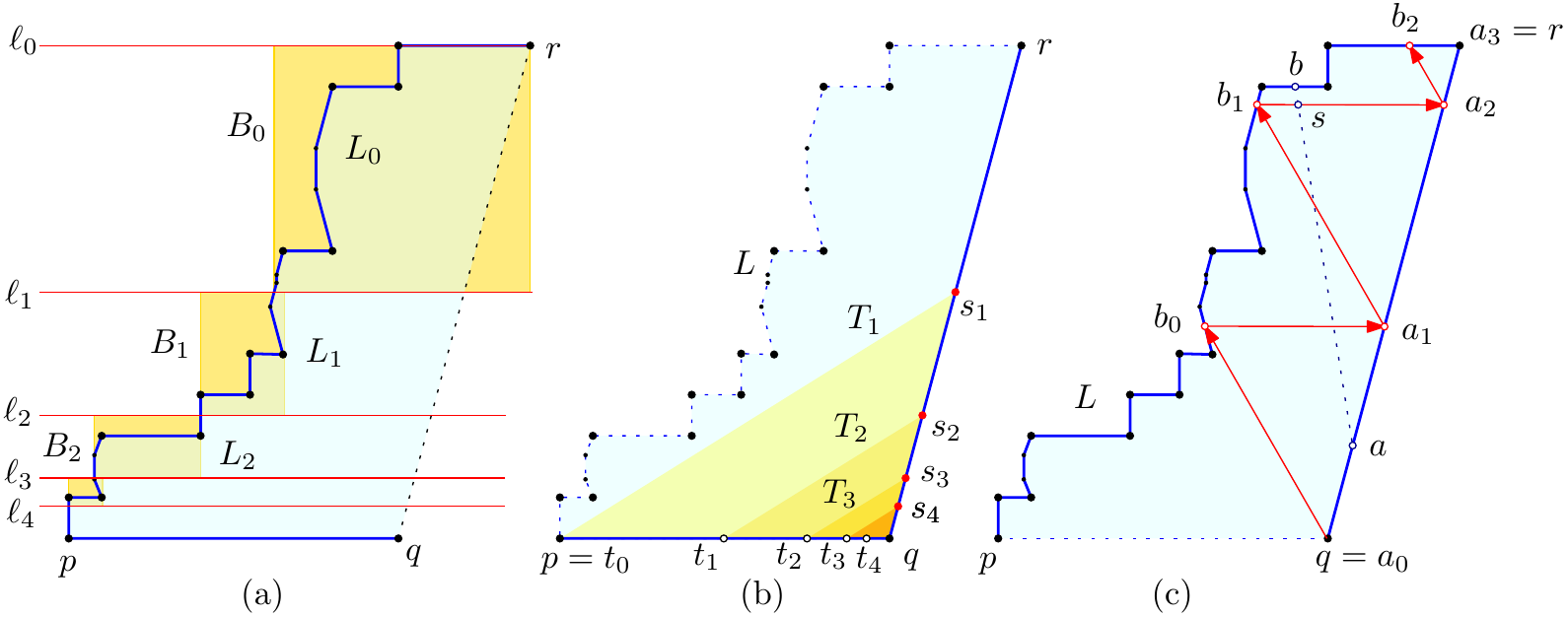}
 \caption{(a) \textsf{SLT}s between the horizontal base $pq$ and $L$.
 (b) \textsf{SLT}s between points $s_i\in qr$ and segments $t_{i-1}q\subset pq$
 (c) \textsf{SLT}s between the right side $qr$ and $L$.}
    \label{fig:fuzzy-stairs}
\end{figure}

\smallskip\noindent\textbf{Case~1: Directional spanner between $pq$ and $L$.}
Refer to Fig.~\ref{fig:fuzzy-stairs}(a).
For every nonnegative integer $i\in \mathbb{N}$, let $\ell_i:y=h/2^i$. Let $L_i$ be the portion of $L$ on or below $\ell_i$ and strictly above $\ell_{i+1}$; see Fig.~\ref{fig:fuzzy-stairs}(a).
We have partitioned $L$ into subpaths $L=\bigcup_{i\geq 0}L_i$, hence $\hper(L)=\sum_{i\geq 0}\hper(L_i)$. Denote by $B_i$ the axis-parallel bounding box of $L_i$, then $\wdth(B_i)\leq \hper(L_i)+\sqrt{\eps}h/2^i$.
For every $i\geq 0$, where $L_i\cap S\neq \emptyset$,
we use Lemma~\ref{lem:combine2} to construct directional
$(1+\eps)$-spanners between $pq$ and $L_i$. The total
weight of these spanners is $O(\eps^{-1/2}\wdth(B_i))=O(\eps^{-1/2}\hper(L_i)+h/2^i)$.
Summation over all $i\geq 0$ yields $O(\eps^{-1/2}\hper(L)+h)=O(\eps^{-1/2}\hper(F))$.

\smallskip\noindent\textbf{Case~2: Directional spanner between $pq$ and $qr$.}
Refer to Fig.~\ref{fig:fuzzy-stairs}(b).
For all $i\in \mathbb{N}$, let $s_i\in pr$ such that $y(s_i)=h/2^i$, let $t_i\in pq$ such that $\|t_iq\|=\|pq\|/2^i$. For all $i\geq 1$, let $T_i$ be a \textsf{SLT} between $s_i$ and the points in the horizontal segment $t_{i-1}q$. Let $G_2$ be the union of $pq\cup qr$, and the \textsf{SLT}s $T_i$ for all $i\geq 1$, for which the interior of segment $t_{i-1}q$ contains any point in $S$.
The combined weight of the \textsf{SLT}s is $\sum_{i\geq 1} O(2^{-i}\eps^{-1/2}\|pq\|)\leq O(\eps^{-1/2}\hper(F))$, as required.

If a point $a\in S$ is on the left side of $R$ at or above $s_i$, and $b$ is at the bottom side of $R$, then the constraint on $\mathrm{dir}(ab)$ implies that $b\in t_{i-1}q$. Consequently, we can construct an $ab$-path by a vertical segment from $a$ to $s_i$, followed by a path from $s_i$ to $b$ in the \textsf{SLT} $T_i$.

\smallskip\noindent\textbf{Case~3: Directional spanner between $qr$ and $L$.}
We reduce this case to the previous two cases; refer to Fig.~\ref{fig:fuzzy-stairs}(c).
We subdivide $P$ by a $qr$-path constructed recursively as follows.
Initially, we set $i=0$ and $a_0=q$. While $a_i\neq r$, we construct
point $b_i\in L$ such that $\slope(a_ib_i)=-\eps^{-1/2}$;
and then construct $a_{i+1}\in qr$ such that $b_ia_{i+1}$ is horizontal.
Since $P$ has finitely many vertices, the algorithm terminates
with $r=a_q$ for some integer $q\geq 1$.

For every $i=0,\ldots , q-1$, we construct the following geometric graph.
The graph includes the subpath of $L$ between $b_i$ and $b_{i+1}$, denoted $L_i$.
We also include the line segment $b_ia_{i+1}$ of weight $\|b_ia_{i+1}\|$.
Between $a_ia_{i+1}$ and $b_ia_{i+1}$,  we construct a geometric series of \textsf{SLT}s
of total weight $O(\hght(a_ia_{i+1}))$, similar to Case~2 above.
Between $b_ia_{i+1}$ and $L_i$, we construct \textsf{SLT}s of total weight $O(\eps^{-1/2}\hper(L_i))$,
similar to  Case~1 above.

\smallskip\noindent\emph{Lightness analysis in Case~3.}
Note that the triangles $\Delta(a_ib_ia_{i+1})$ are similar. Since $\slope(a_ia_{i+1})=\Lambda\eps^{-1/2}$ and
$\slope(a_ib_i)=-\eps^{-1/2}$, then we have $\hght(a_ib_i)\leq \Lambda<\wdth(a_ia_{i+1})$, and
$\|b_ia_{i+1}\|\leq (\Lambda+1)\wdth(a_ia_{i+1})$.
The overall weight of the new edges for all $i=0,\ldots , q-1$ is
\begin{align*}
&\sum_{i=1}^{q-1}\left( \|b_ia_{i+1}\| + O(\hght(a_ia_{i+1}))+O(\eps^{-1/2}\hper(L_i))\right)\\
=&O\left((\Lambda+1)\wdth(qr) + \hght(qr) + \eps^{-1/2}\hper(L) \right)\\
=&O\left(\eps^{-1/2}\wdth(qr) + \eps^{-1/2}\hper(L) \right)
=O\left(\eps^{-1/2}\hper(F)\right).
\end{align*}

\smallskip\noindent\emph{Stretch analysis in Case~3.}
Let $a\in qr$ and $b\in L$ such that $\mathrm{dir}(ab)\in D$.
Assume that $b\in L_i$ for some $i\in \{0,\ldots , q\}$.
Since $\mathrm{dir}(ab)\in D$, this implies that $a\in a_0a_{i+1}$,
in particular $y(a)\leq y(b_i)=y(a_{i+1})$.
Let $s=ab\cap b_ia_{i+1}$; see Fig.~\ref{fig:fuzzy-stairs}(c).
If $a\in a_ia_{i+1}$, then we find a path $P_{ab}$ as
a concatenation of a path from $a$ to $s$ using the \textsf{SLT}s
in the triangle $\Delta(a_ib_ia_{i+1})$, and a path from
$s$ to $b$ using the \textsf{SLT}s between $b_ia_{i+1}$ and $L_i$.
The analysis of Cases~1--2 above implies that
$\|P_{ab}\|\leq (1+\eps)(\|as\|+\|sb\|)=(1+\eps)\|ab\|$.
If $a$ is below point $a_i$, then we construct an $ab$-path $P_{ab}$
as a concatenation of edge $aa_i$, following by a path from
$a_i$ to $b$ via $s$ as in the previous case.
Every edge of the path $aa_isb$ has a direction in the interval $D$,
hence $\|aa_i\|+\|a_is\|+\|sb\|\leq (1+\eps)\|ab\|$ by Lemma~\ref{lem:angle2}.
The \textsf{SLT}s contain paths that approximate $a_is$ and $sb$, resp.,
within a $1+O(\eps)$ factor. Overall, we have
$\|P_{ab}\|\leq (1+O(\eps))\|ab\|$.
\end{proof}

Corollary~\ref{cor:dir3} and Lemma~\ref{lem:dir1} jointly imply Lemma~\ref{lem:hist}.

\histlemma*

This completes all components needed for Theorem~\ref{thm:UB}.

\section{Conclusion and Outlook}
\label{sec:cons}

We have proved a tight upper bound of $O(\eps^{-1})$ on the lightness of Euclidean Steiner $(1+\eps)$-spanners in the plane. That is, for every finite set $S\subset \mathbb{R}^2$, there is a Euclidean Steiner $(1+\eps)$-spanner of weight $O(\eps^{-1}\,\|\MST(S)\|)$. Our proof is constructive, but we do not control the number of Steiner points. This immediately raises the question about the optimum number of Steiner points: What is the minimum sparsity of a Euclidean Steiner $(1+\eps)$-spanner of weight $O(\eps^{-1}\|\MST(S)\|)$ that can be attained for all finite set of points in $\mathbb{R}^2$?

Planarity is an important aspect of any geometric networks.
Therefore, it is desirable to construct Euclidean $(1+\eps)$-spanners that are plane, i.e., no two edges of the spanner cross. Any Steiner spanner can be turned into a plane spanner (planarized), with the same weight and the same spanning ratio between the input points, by introducing Steiner points at all edge crossings. However, planarization may substantially increase the number of Steiner points. Bose and Smid~\cite[Sec.~4]{BoseS13} note that Arikati et al.~\cite{ArikatiCCDSZ96} constructed a Euclidean plane $(1+\eps)$-spanner with $O(\eps^{-4} n)$ Steiner points for $n$ points in $\mathbb{R}^2$; see also~\cite{MaheshwariSZ08}. Borradaile and Eppstein~\cite{BorradaileE15} improved the bound to $O(\eps^{-3} n\log \eps^{-1})$ in certain special cases where all Delaunay faces of the point set are fat. It remains an open problem to find the optimum dependence of $\eps$ for plane Steiner $(1+\eps)$-spanners; and for plane Steiner $(1+\eps)$-spanners of lightness $O(\eps^{-1})$.

\bibliographystyle{plainurl}
\bibliography{spanner}

\end{document}